\documentclass[11pt,a4paper]{article}
\usepackage{hyperref}
\usepackage{amsmath,amsfonts,amssymb,amsthm,mathtools}
\usepackage{authblk}
\usepackage{latexsym,graphicx}
\usepackage{dsfont}
\usepackage{amscd}
\usepackage[dvipsnames]{xcolor}
\usepackage{extarrows} 
\usepackage{mathtools}
	\numberwithin{equation}{section} 
\usepackage{yfonts}
\usepackage{authblk}

\usepackage{tikz}
	 \usetikzlibrary{calc} 
	\usetikzlibrary{matrix,arrows,decorations.pathmorphing} 

\allowdisplaybreaks

\def\Xint#1{\mathchoice
{\XXint\displaystyle\textstyle{#1}}%
{\XXint\textstyle\scriptstyle{#1}}%
{\XXint\scriptstyle\scriptscriptstyle{#1}}%
{\XXint\scriptscriptstyle\scriptscriptstyle{#1}}%
\!\int}
\def\XXint#1#2#3{{\setbox0=\hbox{$#1{#2#3}{\int}$}
\vcenter{\hbox{$#2#3$}}\kern-.5\wd0}}

\def\pvint{\Xint-}

\newcommand{\aR}{a^\mathrm{R}}
\newcommand{\taR}{\tilde{a}^\mathrm{R}}
\newcommand{\aI}{a^{\mathrm{I}}} 

\newcommand{\ssum}[3]{\sum_{{#1}\neq {#2}}^{#3}}
\newcommand{\pprod}[3]{\prod_{{#1}\neq {#2}}^{#3}}
\newcommand{\sssum}[4]{\sum_{{#1}\neq {#2},{#3}}^{#4}}

\newcommand{\ee}{{\rm e}}
\newcommand{\ii}{{\rm i}}

\newcommand{\R}{{\mathbb R}}
\newcommand{\C}{{\mathbb C}}
\newcommand{\Z}{{\mathbb Z}}

\newcommand{\cN}{{\mathcal  N}}

\newtheorem{theorem}{Theorem}[section]
\newtheorem{lemma}[theorem]{Lemma}
\newtheorem{proposition}[theorem]{Proposition}

\theoremstyle{definition}

\theoremstyle{remark}
\newtheorem{remark}[theorem]{Remark}

\newcommand{\bb}{\mathbf{b}}
\newcommand{\bc}{\mathbf{c}}
\newcommand{\be}{\mathbf{e}}
\newcommand{\bm}{\mathbf{m}}
\newcommand{\bn}{\mathbf{n}}
\newcommand{\bmb}{\mathbf{m}_0}
\newcommand{\bs}{\mathbf{s}}
\newcommand{\bt}{\mathbf{t}}
\newcommand{\bx}{\mathbf{x}}
\newcommand{\by}{\mathbf{y}}
\newcommand{\bv}{\mathbf{v}}
\newcommand{\bw}{\mathbf{w}}
\newcommand{\bz}{\mathbf{z}}
\newcommand{\bS}{\mathbf{S}}

\newcommand{\re}{\mathrm{Re}\,}
\newcommand{\im}{\mathrm{Im}\,}

\title{\Large{Multi-solitons of the half-wave maps equation and Calogero-Moser spin-pole dynamics}}
\date{\vspace{-0.5cm}\small\today\vspace{-0.5cm}}

\author[1]{Bjorn K. Berntson}
\author[2]{Rob Klabbers}
\author[3]{Edwin Langmann}

\affil[1]{Department of Mathematics, KTH Royal Institute of Technology, SE-100 44 Stockholm, Sweden}
\affil[2]{\textsc{Nordita}, KTH Royal Institute of Technology and Stockholm University, SE-106 91 Stockholm, Sweden}
\affil[3]{Department of Physics, KTH Royal Institute of Technology, SE-106 91 Stockholm, Sweden}

\vspace{2mm}

\date{\today}

\begin{document}

\begin{flushright}

	\footnotesize
	\textsc{nordita} 2020-069
\end{flushright}

\bigskip

{\let\newpage\relax\maketitle}

\maketitle

\newcommand{\E}[1]{{\color{blue}{#1}}}

\begin{abstract}
We consider the half-wave maps (HWM) equation which provides a 
 continuum description of the classical Haldane-Shastry spin chain on the real line. 
We present exact multi-soliton solutions of this equation. 
Our solutions describe solitary spin excitations that can move with different velocities and interact in a non-trivial way. 
We make an ansatz for the solution allowing for an arbitrary number of solitons, each described by a pole in the complex plane and a complex spin variable, 
and we show that the HWM equation is satisfied if these poles and spins evolve according to the dynamics of an exactly solvable spin Calogero-Moser (CM) system with certain constraints on initial conditions.
We also find first order equations providing a B\"acklund transformation of this spin CM system, generalize our results to the periodic HWM equation, 
and provide plots that visualize our soliton solutions. 
\end{abstract} 

\section{Introduction} 

One powerful method to  describe systems with a large number of interacting degrees of freedom is to take a hydrodynamic limit \cite{spohn1991}. 
Since hydrodynamic limits are difficult to perform in practice, one often has to resort to phenomenological hydrodynamic equations. 
However, within the class of integrable systems, there are important exceptions: examples for which precise hydrodynamic descriptions are known.
Calogero-Moser (CM) systems are prominent such examples which have hydrodynamic descriptions closely related to soliton equations of Benjamin-Ono type \cite{polychronakos1995,stone2008,abanov2009,kulkarni2017,berntson2020}.
The relation between integrable systems and hydrodynamic equations has recently received considerable interest in the context of non-equilibrium physics; see \cite{doyon2019,spohn2019} and references therein.

In this paper we present basic results for a soliton equation that was derived as a 
 continuum limit of a classical version of the Haldane-Shastry spin chain \cite{zhou2015} and which is known as the half-wave maps (HWM) equation \cite{lenzmann2018,lenzmann2018b}. 
The HWM equation describes the time evolution of a spin density in one dimension represented by a $S^2$-valued function $\bm(x,t)$ depending on a spatial variable $x\in\R$ and time $t \in \R$,\footnote{We find it convenient to treat the time evolution both to $t>0$ and $t<0$ on equal footing.} and it is given by\footnote{By $\bm\in S^2$ we mean $\bm\in\R^3$ with the constraint 
$\bm^2=1$; $\wedge$ is the usual cross product of three-vectors.}
\begin{equation} 
\label{hwm}
\bm_t = \bm\wedge H\bm_x
\end{equation} 
with $\bm_t=\frac{\partial}{\partial t}\bm$ etc., and $H$ the Hilbert transform: 
\begin{equation} 
\label{hilbert_line}
(Hf)(x) \coloneqq \frac1\pi \pvint_{\R} \frac{f(x')}{x'-x}\,  \mathrm{d}x'
\end{equation} 
with $\pvint$ a Cauchy principal value integral (for simplicity and to be specific, we restrict our discussion here to the HWM equation on the real line; however, as discussed further below, we also have results for the  periodic HWM equation). The Haldane-Shastry spin chain can be obtained from a spin CM system in a limit where the particle positions are frozen to a lattice \cite{polychronakos1993} and, for this reason,  we regard the HWM equation as limiting case of a hydrodynamic equation describing a spin CMS system.
 
Since the HWM equation arises as a continuum limit of an integrable system, one expects that it is integrable. However, while some results in this direction are known, there are still gaps in our understanding of this equation. More specifically, while a Lax pair of the HWM equation was recently found in  \cite{gerard2018}, only a restricted class of soliton solutions is known \cite{zhou2015,lenzmann2018,lenzmann2018b}. This class consists of solutions with an arbitrary number, $N$, of solitons, all moving with the same constant velocity and without interactions. Numerical results \cite{zhou2015} suggest that there exist more general soliton solutions where the individual solitons move with different velocities and interact in a non-trivial way. Our main result is exact analytic formulas for such general multi-soliton solutions of the HWM equation.

We now describe our main result (a precise formulation is given in Theorem \ref{thm}). 
We make the ansatz
\begin{equation}\label{ansatz}
\bm(x,t)=\bmb+\ii\sum_{j=1}^N \frac{\bs_j(t)}{x-a_j(t)}-\ii\sum_{j=1}^N \frac{\bs_j(t)^*}{ x-a_j(t)^*}
\end{equation}
with $*$ complex conjugation, $\bmb\in S^2$ describing an arbitrary vacuum solution (i.e., a solution that is constant in space and time), $a_j(t)$ poles in the upper half complex plane $\C_+$, and $\bs_j(t)$ spin variables with values in $\C^3$, and we show that this ansatz gives a solution of the HWM equation \eqref{hwm}--\eqref{hilbert_line} provided the following equations hold true:
\begin{equation}\label{resultsj1} 
\dot{\bs}_j=- 2 \ssum{k}{j}{N} \frac{\bs_j\wedge\bs_k}{(a_j-a_k)^2},
\end{equation} 
\begin{equation}\label{resultaj1}
\dot{a}_j\bs_j=-\bs_j \wedge    \Bigg( \ii\bmb-\ssum{k}{j}{N}\frac{\bs_k}{a_j-a_k}+\sum_{k=1}^N\frac{\bs_k^*}{a_j-a_k^*} \Bigg),
\end{equation}
\begin{equation}\label{constraint}
\bs_j^2=0, \quad \bs_j\cdot \bigg(\ii \bmb-\ssum{k}{j}{N} \frac{\bs_k}{a_j-a_k}+\sum_{k=1}^N \frac{\bs_k^*}{a_j-a_k^*} \bigg)=0 
\end{equation}
for $j=1,\ldots, N$. 
We also show that equations \eqref{resultsj1}--\eqref{constraint} provide a B\"acklund transformation for a known spin CM system in the sense of Wojciechowski \cite{wojciechowski1982}: they imply
\begin{equation}
\label{eq:ddota_rational1}
\ddot{a}_j=4\ssum{k}{j}{N} \frac{\bs_j\cdot\bs_k}{(a_j-a_k)^3}    
\end{equation}
for $j=1,\ldots,N$,  and \eqref{resultsj1} and \eqref{eq:ddota_rational1} are the equations of motion of an exactly solvable spin CM system solved in \cite{gibbons1984generalisation,krichever1995spin}. 

It is important to note that, if \eqref{resultaj1} and \eqref{constraint} hold true at initial time $t=0$, they are fulfilled for all times $t\in\R$ provided $\bs_j(t)$ and $a_j(t)$ time evolve according to \eqref{resultsj1} and \eqref{eq:ddota_rational1}; for this reason, \eqref{resultaj1} and \eqref{constraint} are constraints on initial conditions: the dynamics is given by the spin CM system. Moreover, for fixed $N$ and $\bmb$, the constraints in \eqref{constraint} at time $t=0$ allow for solutions parametrized by $4N$ real parameters: the initial pole positions $a_{j,0} \coloneqq a_{j}(0)$ in the upper half complex plane, and directions $\mathbf{n}_{j}\in S^2$ determining the initial complex spins $\bs_{j,0} \coloneqq \bs_{j}(0)$, for $j=1,\ldots,N$. As we will show in Section~\ref{sec:travel}, the previously known traveling wave solutions \cite{zhou2015,lenzmann2018,lenzmann2018b} correspond to the special case where $\mathbf{n}_{j}=\mathbf{n}$ is the same for all $j$. We also present different methods to find initial data satisfying the constraints as well as plots of our solutions which demonstrate that our soliton solutions can describe complicated spin interactions. 

For simplicity, we restricted our discussion above to the HWM equation on the real line. There is also a periodic version of the HWM equation which is integrable as well \cite{zhou2015,lenzmann2018,gerard2018}. 
Our main result stated above straightforwardly generalizes to the periodic case.

In the main body of this paper, we restrict ourselves to $S^2$-valued solutions, which are the most relevant for physics. However, our results straightforwardly generalize to the case where $\bm$ is $\C^3$-valued with $\bm^2$ fixed to an arbitrary constant complex value:
all the proofs in the Appendices are done for this more general case.

The plan of this paper is as follows. In Section \ref{sec2} we derive our $N$-soliton solutions of the HWM equation on the real line, explaining the key steps and deferring computational details to appendices. The generalizations of these results to the periodic case can be found in Section \ref{sec3}. 
In Section \ref{sec:constraints}, we discuss several methods to find initial data of multi-solitons and present visualizations of examples of our solutions. We end with concluding remarks in Section \ref{sec:discussion}. 
Details of the proofs can be found in the Appendices. \\

\noindent {\bf Notation:} We write $\sum_{k\neq j}^N$ short for sums $\sum_{\genfrac{}{}{0pt}{}{k=1}{k\neq j}}^N$ etc.  
Vectors $\bm$ in $\R^3$ or $\C^3$ are written as $\bm=(m_1,m_2,m_3)$, and $\bm\wedge\bn =(m_2n_3-m_3n_2,m_3n_1-m_1n_3,m_1n_2-m_2n_1)$. 
 We denote as $\C_+$ the complex upper half plane: $\C_+\coloneqq\{ z\in\C  |  \im z>0\}$, and similarly, $\C_-$ is the complex lower half plane.  

\section{Multi-soliton solutions on the real line}
\label{sec2}
We derive the multi-soliton solutions of the HWM equation on the real line governed by the spin-pole dynamics described in the introduction (Section~\ref{sec:derivation1}), and we show that this dynamics can be derived from a known  spin CM system (Section~\ref{sec:CS_dynamics}). We also explain how to recover the known traveling wave solutions as a special case of our multi-soliton solutions (Section~\ref{sec:travel}), and we discuss the solution of the constraints on initial conditions (Section~\ref{sec:constraints1}). 

\subsection{Spin-pole dynamics}
\label{sec:derivation1} 
We find it convenient to use the notation 
\begin{equation} 
\label{alphaV} 
\alpha(x) \coloneqq \frac1x ,\quad V(x)\coloneqq \frac1{x^2}
\end{equation} 
and 
\begin{equation}
\label{absr}
(a_j,\bs_j,r_j)=\begin{cases}
(a_j,\bs_j,+) & (j=1,\ldots,N), \\
(a_{j-N}^*,\bs_{j-N}^*,-) & (j=N+1,\ldots 2N). 
\end{cases}
\end{equation}
Note that $V(x)$ is the famous interaction potential in the rational CM model, and $\alpha(x)$ is the associated special function satisfying 
$V(x)=\alpha(x)^2=-\alpha'(x)$ and a functional identity given in \eqref{Id} below \cite{calogero1975}. 

Using this, the pole ansatz in \eqref{ansatz} can be written as 
\begin{equation}\label{ansatzproof1}
\bm(x,t)=\bmb+\ii\sum_{j=1}^{2N} r_j\bs_j(t) \alpha(x-a_j(t))
\end{equation}
where $\bmb$ is an arbitrary constant vector in  $\R^3$ satisfying  $\bmb^2=1$. 
Note that, since we assume that the poles $a_j(t)$ for $j=1,\ldots,N$ are in the complex upper half-plane, $r_j$ equals the sign of the imaginary part of $a_j(t)$ for $j=1,\ldots,2N$.  

In general, the pole ansatz in \eqref{ansatz} gives a function $\bm(x,t)\in \R^3$ and thus, to find a solution of the HWM equation,  it is important to find conditions that constrain to $\bm(x,t)^2=1$. We therefore compute, using \eqref{ansatzproof1}, 
\begin{align*}
\bm^2= \bmb^2+2\ii  \bmb\cdot \sum_{j=1}^{2N}r_j\bs_j\alpha(x-a_j)-\sum_{j=1}^{2N}\sum_{k=1}^{2N}r_jr_k (\bs_j\cdot\bs_k)\alpha(x-a_j)\alpha(x-a_k) 
\end{align*}
(here and in the following, we write $\bm$ short for $\bm(x,t)$, and we suppress the time dependence of $\bs_j$ and $a_j$, to simplify notation). 
Inserting $\bmb^2=1$ and evaluating the double sum using $\alpha(x-a_j)^2=V(x-a_j)$ for $k=j$ and  
\begin{equation}
\label{Id} 
\alpha(x-a_j)\alpha(x-a_k)=\alpha(a_j-a_k)\big(\alpha(x-a_j)-\alpha(x-a_k)\big)
\end{equation}
for $k\neq j$,  we find
\begin{align*}
\bm^2=
1 -\sum_{j=1}^{2N} \bs_j^2 V(x-a_j) + 2\sum_{j=1}^{2N}\alpha(x-a_j)r_j\bs_j\cdot \bigg( \ii\bmb -\ssum{k}{j}{2N}r_k\bs_k\alpha(a_j-a_k) \bigg) 
\end{align*}
(computational details are given in Appendix~\ref{lemma1proof}). 
Because the functions $\alpha(x-a_j)$ and $V(x-a_j)$ are linearly independent, their coefficients must vanish if $\bm^2=1$: 
\begin{equation}\label{constraint22}
\bs_j^2=0,\quad \bs_j\cdot \bigg( \ii\bmb -\ssum{k}{j}{2N}r_k\bs_k\alpha(a_j-a_k) \bigg)=0
\end{equation} 
for $j=1,\ldots,2N$. Inserting \eqref{alphaV}--\eqref{absr}, one obtains the conditions in \eqref{constraint} for $j=1,\ldots,N$ (and the complex conjugate thereof for $j=N+1,\ldots,2N$). 
To summarize:  {\em The conditions in \eqref{constraint} are necessary and sufficient for the pole ansatz \eqref{ansatz} to satisfy $\bm^2=1$.} 

Because the HWM equation \eqref{hwm}--\eqref{hilbert_line}  is length-preserving: 
\begin{equation} 
\label{eq:length_preserving}
\partial_t\bm^2=2\bm_t\cdot\bm=2(\bm\wedge H\bm_x)\cdot\bm=0,
\end{equation} 
choosing initial values for $a_j$ and $\bs_j$ that satisfy \eqref{constraint} at time $t=0$ is sufficient to guarantee $\bm\in S^2$ at future time. 
We suppose these initial values have been appropriately chosen and seek the differential equations governing their evolution. 

We substitute the pole ansatz \eqref{ansatzproof1} into \eqref{hwm}--\eqref{hilbert_line}  to find, on the left-hand side,
\begin{equation}\label{mtrational}
\bm_t=\ii\sum_{j=1}^{2N} r_j\big(\dot{\bs}_j \alpha(x-a_j)+\bs_j \dot{a}_j V(x-a_j)\big).
\end{equation}
To compute the right-hand side of \eqref{hwm}--\eqref{hilbert_line}, we recall that the (boundary values of) functions analytic in the upper and lower half-planes are eigenfunctions of the Hilbert transform with eigenvalues $\pm \ii$, respectively. It follows that  $H\alpha(x-a_j)=-\ii r_j \alpha(x-a_j)$ and
\begin{align*}
H\bm_x=\partial_x (H\bm)= \ii\partial_x \bigg( \sum_{j=1}^{2N} (-\ii r_j) r_j\bs_j \alpha(x-a_j)       \bigg) =-\sum_{j=1}^{2N}  \bs_j V(x-a_j),
\end{align*}
using $\alpha'(x-a_j)=-V(x-a_j)$ and the fact that $H$ commutes with $\partial_x$. 
Inserting this into the right-hand side of \eqref{hwm}--\eqref{hilbert_line}  and using the identity obtained by differentiating the one in \eqref{Id} with respect to $a_k$ we find, by straightforward computations, 
\begin{align}
\bm\wedge H\bm_x=
&-\ii \sum_{j=1}^{2N}  V(x-a_j)\bs_j\wedge \bigg( \ii \bmb - \sum_{k=1}^{2N}r_k \bs_k \alpha(a_j-a_k)       \bigg) \label{mhmxrational}\\
&-\ii\sum_{j=1}^{2N}\ssum{k}{j}{2N}  \alpha(x-a_j)(r_j+r_k) (\bs_j\wedge\bs_k) V(a_j-a_k) \nonumber
\end{align}
(computational details can be found in Appendix~\ref{theorem1proof}). 
Since the functions $\alpha(x-a_j)$ and $V(x-a_j)$ are linearly independent, \eqref{mtrational} and \eqref{mhmxrational} imply that \eqref{hwm}--\eqref{hilbert_line}  is fulfilled if and only if 
\begin{subequations}\label{theorem1eom1}
\begin{align}
\dot{\bs}_j=&-\ssum{k}{j}{2N}(1+r_jr_k)(\bs_j\wedge\bs_k)V(a_j-a_k), \label{sjeq} \\
\dot{a}_j \bs_j=& -r_j\bs_j \wedge \bigg( \ii \bmb-\ssum{k}{j}{2N} r_k\bs_k \alpha(a_j-a_k)  \bigg). \label{ajeq}
\end{align}
\end{subequations}
By \eqref{alphaV}--\eqref{absr}, \eqref{theorem1eom1} is equivalent to \eqref{resultsj1}--\eqref{resultaj1}.   To summarize: {\em If the pole ansatz in \eqref{ansatz} satisfies $\bm^2=1$ at initial time $t=0$, and if the spins $\bs_j$ and poles $a_j$ time evolve according to  \eqref{resultsj1}--\eqref{resultaj1} and are such that $a_j\in\C_+$ for $j=1,\ldots,N$, then the function $\bm$ in \eqref{ansatz} is a solution of the HWM equation \eqref{hwm}--\eqref{hilbert_line} satisfying $\bm^2=1$ for all times $t\in\R$.}  

It is not obvious but true that \eqref{ajeq} is a well-defined differential equation determining the time evolution of the poles $a_j$: as explained in Appendix~\ref{app:consistency}, \eqref{ajeq} can be consistently reduced to the following scalar equation, 
\begin{equation}  
\dot{a}_j =  r_j \frac{\bs_j\wedge \bs_j^*}{\bs_j\cdot\bs_j^*} \cdot \bigg( \ii \bmb-\ssum{k}{j}{2N} r_k\bs_k \alpha(a_j-a_k)  \bigg)
\end{equation} 
(note that this equation is obtained from \eqref{ajeq} by taking the dot product  with $\bs_j^*$). 
Another way to see this is to differentiate \eqref{ajeq} with respect to time, and to simplify the resulting equation by a straightforward but tedious computation  using \eqref{constraint22} and \eqref{theorem1eom1}; this gives 
\begin{equation} 
\ddot{a}_j  = -\ssum{k}{j}{2N} (1+r_jr_k)(\bs_j\cdot\bs_k) V'(a_j-a_k)
\end{equation} 
(computational details are given in Appendix~\ref{proofddotaj}); by \eqref{alphaV}--\eqref{absr}, this is equivalent to \eqref{eq:ddota_rational1}. 
Thus, rather than solving \eqref{resultsj1}--\eqref{resultaj1} with initial conditions satisfying \eqref{constraint}, we can interpret \eqref{resultaj1} as a further constraint on initial conditions and determine the time evolution of $\bs_j(t)$ and $a_j(t)$ by solving \eqref{resultsj1} and \eqref{eq:ddota_rational1}. As discussed in Section~\ref{sec:CS_dynamics}, the latter two equations define the dynamics of a spin CM system which is known to be integrable. 

We summarize our findings as follows. 

\begin{theorem} 
\label{thm} 
For arbitrary $\bmb\in S^2$, $N\in\Z_{\geq 1}$, $\bs_{j,0}\in\C^3$ and $a_{j,0}\in\C_+$  such that 
\begin{equation} 
\label{constraintthm} 
\bs_{j,0}^2=0,\quad \bs_{j,0}\cdot\bigg(\ii\bmb - \ssum{k}{j}{N}\frac{\bs_{k,0}}{a_{j,0}-a_{k,0}} +  \sum_{k=1}^N\frac{\bs^*_{k,0}}{a_{j,0}-a_{k,0}^*} \bigg)=0  
\end{equation} 
for $j=1,2,\ldots,N$, let $\bs_j(t)$ and $a_j(t)$ be solutions of the following system of equations,  
\begin{subequations} 
\label{dyn}
\begin{align} 
\dot\bs_j(t) 
& = -2\ssum{k}{j}{N}\frac{\bs_j(t)\wedge\bs_k(t)}{\left(a_j(t)-a_k(t)\right)^2}, \label{sjdotthm}\\
\ddot a_j(t) & = 4\ssum{k}{j}{N}\frac{\bs_j(t)\cdot\bs_k(t)}{\left(a_j(t)-a_k(t)\right)^3} \label{ajddotthm},
\end{align} 
\end{subequations} 
with initial conditions $\bs_j(0)=\bs_{j,0}$, $a_j(0)=a_{j,0}$, and 
\begin{equation} 
\dot a_j(0) = \frac{\bs_{j,0}\wedge \bs_{j,0}^*}{\bs_{j,0}\cdot\bs_{j,0}^*}\cdot\bigg( \ii\bmb   - \ssum{k}{j}{N} \bs_{k,0}\frac1{a_{j,0}-a_{k,0}}
+ \sum_{k=1}^{N}\bs^*_{k,0}\frac1{a_{j,0}-a^*_{k,0}}  \bigg)
\end{equation} 
for $j=1,\ldots,N$. Then, for all $t\in\R$ such that  (i) $a_j(t)\in\C_+$ for all $j=1,\ldots,N$,  (ii) $\bm(x,t)$ is differentiable with respect to $x$ and $t$, $\bm(x,t)$ in \eqref{ansatz} is an exact solution of the HWM equation \eqref{hwm}--\eqref{hilbert_line}  such that $\bm(x,t)^2=1$. 
\end{theorem} 

\begin{remark} 
\label{rem:cross}
We restrict the result to times such that $a_j(t)\in\C_+$  for all $j$ since, if one of the poles crosses the real line at some time $t=t_0\gtrless 0$, then our proof breaks down for $t\gtrless t_0$. 
In our numerical experiments we never saw this happen, and we believe that it cannot happen; it would be interesting to prove this.  
\end{remark} 
\begin{remark}
\label{rem:cusp} 
For simplicity, we exclude the possibility of non-differentiable solutions $\bm(x,t)$. 
However, we found examples of initial data $\bm(x,0)$ of the form \eqref{ansatz} with one and two cusps.
Moreover, it is conceivable that certain differentiable initial data $\bm(x,0)$ of the form  \eqref{ansatz} develop cusps in either $x$ or $t$ as time evolves; in fact, we observed
such cusps in numerical experiments; see Fig.~\ref{fig:three_soliton}. While we specifically exclude such cusps in the above theorem for simplicity, it is tempting to speculate that profiles $\bm(x,t)$ with cusps are in fact weak solutions of \eqref{hwm} in analogy with the peakon solutions originally discovered for the Camassa-Holm equation \cite{camassa1993}. We leave this interesting question to future work.
\end{remark} 

\subsection{Relation to rational spin Calogero-Moser system}
\label{sec:CS_dynamics}
We show that the equation determining the spin-pole dynamics in \eqref{dyn} are identical with the equations of motion of the rational spin CM system due to Gibbons and Hermsen \cite{gibbons1984generalisation} in a special case \cite{krichever1995spin}.

Consider the Hamiltonian  
\begin{equation}
\label{eq:ratCMham}
H_{\mathrm{CM}} = \frac{1}{2}\sum_{j=1}^N p_j^2+ \sum_{1\leq j<k\leq N}  \frac{2  \bS_j\cdot \bS_k}{(q_j-q_k)^2}
\end{equation}
that drives the classical time evolution of the variables  $q_j$, $p_j$, and $\bS_j\coloneqq (S_j^1,S_j^2,S_j^3)$ according to the Poisson brackets
\begin{equation}
\label{eq:poisson_brackets_rat}
\{q_j, p_k \} = \delta_{jk}  , \quad \{S_j^a, S_k^b \} =   \delta_{jk}  \epsilon_{abc}   S_j^c
\end{equation}
for $j,k=1,\ldots,N$ and $a,b=1,2,3$, where $\epsilon_{abc}$ is the completely antisymmetric symbol with $\epsilon_{123}=1$. The equations of motion derived from this Hamilton system are
\begin{subequations} 
\label{dyn2}
\begin{align} 
\dot\bS_j  & = -2\ssum{k}{j}{N}\frac{\bS_j\wedge\bS_k}{(q_j-q_k)^2}, \label{sjdotthm111}\\
\ddot q_j & = 4 \ssum{k}{j}{N}\frac{\bS_j\cdot\bS_k}{(q_j-q_k)^3} . \label{ajddotthm111}
\end{align} 
\end{subequations} 
We observe that these equations are identical with the time evolution equations in \eqref{dyn} if we identify $q_j$ and $\bS_j$ with $a_j$ and $\bs_j$, respectively.  
However, there is an important difference: the variables $q_j$ and $\bS_j$ here correspond to particle positions on the real line and real spins, respectively, whereas the poles $a_j$ and spin variables $\bs_j$ are complex and thus have no direct physical interpretation. 
Thus, the spin-pole dynamics for the HWM equation we found corresponds to a peculiar complexified version of the spin CM system in  \eqref{eq:ratCMham}--\eqref{eq:poisson_brackets_rat}. 
However, one still can use the known exact solution of this system  \cite{gibbons1984generalisation} to solve the equations in \eqref{dyn}, and we therefore obtained fully analytic multi-soliton solutions of the HWM equation \eqref{hwm}--\eqref{hilbert_line}. 

In the rest of this section we show that the Hamiltonian system \eqref{eq:ratCMham}--\eqref{eq:poisson_brackets_rat} is indeed a special case of the rational spin Calogero-Moser model in \cite{gibbons1984generalisation,krichever1995spin}. 

The rational spin Calogero-Moser model is defined by the Hamiltonian \cite{gibbons1984generalisation}
\begin{equation}
\label{eq:HKBBT}
H_{\mathrm{GH}} \coloneqq  \frac{1}{2} \sum_{j=1}^N p_j^2 +
 \sum_{1\leq j<k\leq N}\frac{(\bv_j\cdot \bw_k)   (\bv_k\cdot \bw_j) }{(q_j -q_k)^2}  ,
\end{equation}
generating the time evolution of $N$ particles with coordinates $q_j$ and  momenta $p_j$, together with two sets of internal degrees of freedom 
given by two $d$-vectors $\bv_j=(v_j^\alpha)_{\alpha=1}^N$ and $\bw_j=(w_j^\alpha)_{\alpha=1}^N$, with $\bv_j\cdot \bw_j\coloneqq \sum_{\alpha=1}^d v_j^\alpha w_j^\alpha$ and the following non-trivial Poisson brackets, 
\begin{equation}
\label{eq:KBBT_Poisson}
\{q_j, p_k \}  = \delta_{jk}  , \quad \{ v_j^\alpha, w_k^{\beta} \} = -\ii \delta_{jk} \delta_{\alpha \beta}  . 
\end{equation}
It is straightforward to check that $\bv_j\cdot \bw_j$ are integrals of motion. The reduction of the system to a submanifold determined by
\begin{equation}
\label{eq:restriction}
\bv_j\cdot \bw_j = d   c  \quad (j=1,\ldots,N)  ,
\end{equation}
for any $c$ is integrable \cite{krichever1995spin}. Define the matrices $M_j$ by their entries $(M_j)_{\alpha \beta} \coloneqq v_j^\alpha w_j^\beta$, such that $(\bv_j\cdot \bw_k)   (\bv_k\cdot \bw_j) = \mathrm{Tr}\left( M_j M_k\right)$ and the restriction \eqref{eq:restriction} becomes
\begin{equation}
\mathrm{Tr}(M_j) = d  c  .
\end{equation}
Let us now specialize to $d=2$ and $c=0$, i.e. the case of two-dimensional vectors and $2\times 2$ traceless matrices. We can decompose any such matrix as $M_j = \bS_j \cdot \boldsymbol{\sigma} $, where $\boldsymbol{\sigma}$ is the vector of Pauli matrices. The coefficients $\bS_j$ can be found using $\bS_j = \frac{1}{2} \mathrm{Tr} ( \boldsymbol{\sigma} M_j)$. We can now rewrite the system in terms of $\bS_j$: the Poisson brackets $ \{S_j^a, S_k^b \} $ following from \eqref{eq:KBBT_Poisson} equal the ones in \eqref{eq:poisson_brackets_rat} and since $\mathrm{Tr}\left( M_j M_k\right) =  2\bS_j \cdot   \bS_k$, the Hamiltonian in \eqref{eq:HKBBT} reduces to the one in \eqref{eq:ratCMham}. 

\subsection{Traveling wave solutions}
\label{sec:travel} 
We discuss how previously known traveling wave solutions of the HWM equation are recovered from our soliton solutions (Section~\ref{sec:derivation}), and we discuss the physical interpretation of one-soliton solutions (Section~\ref{sec:onesoliton}). 

\subsubsection{Derivation}
\label{sec:derivation} 
It is known that the HWM equation admits exact traveling wave solutions \cite{zhou2015,lenzmann2018}
\begin{equation}\label{travelingwave}
\bm(x,t)=\big(\sqrt{1-v^2}\,  \re B(x-vt) ,\mp \sqrt{1-v^2}\,  \im B(x-vt),\mp v\big), \qquad -1<v<1,
\end{equation}
where
\begin{equation}\label{blaschke}
B(z)=\prod_{j=1}^N \frac{z-a_{j,0}}{z-a_{j,0}^*}, \qquad a_{j,0} \in \C_+.
\end{equation}
We show how to recover these solution from the ansatz \eqref{ansatz}. 

We set
\begin{equation}
\bmb= \cos \theta  \be_1+\sin\theta  \be_3,\quad \bs_j= s_j(\be_1\mp \ii \be_2) \quad (j=1,\ldots,N)
\end{equation}
with $s_j\in \C$ and $\be_1=(1,0,0)$, $\be_2=(0,1,0)$, $\be_3=(0,0,1)$. We see that the first constraint in \eqref{constraint} is satisfied: $\bs_j^2=0$, while the second constraint in \eqref{constraint} becomes
\begin{equation}\label{travelingwaveconstraint}
\cos\theta-\sum_{k=1}^N \frac{2\ii s_k^*}{a_j-a_k^*}=0. 
\end{equation}
We have $\dot{\bs}_j=0$ from \eqref{resultsj1}, while \eqref{resultaj1} leads to 
\begin{align*}
\dot{a}_j \bs_j=& s_j (\be_1\mp \ii \be_2)\wedge \Bigg( -\ii ( \cos\theta \be_1+\sin\theta \be_3 ) + \sum_{k=1}^{N}\frac{s_k^*(\be_1\pm\ii\be_2)}{a_j-a_k^*}      \Bigg) \\
=&\mp \sin\theta \bs_j \mp s_j \cos\theta \be_3\pm \sum_{k=1}^N\frac{2\ii s_j s_k^*}{a_j-a_k^*}\be_3 = \mp \sin\theta \bs_j  ,
\end{align*}
using \eqref{travelingwaveconstraint} in the last step. Thus $\bs_j(t)=s_{j,0}(\be_1\mp\ii\be_2)$ and $a_j(t)=a_{j,0}\mp \sin\theta  t$,  and we have found the solution
\begin{align}\label{travelingwavesolution}
\bm(x,t)=&\cos \theta  \be_1+\sin\theta  \be_3+\ii \sum_{j=1}^N \frac{s_{j,0}(\be_1\mp \ii\be_2)}{x-a_{j,0}\pm \sin\theta t }-\ii \sum_{j=1}^N \frac{s_{j,0}^*(\be_1\pm \ii\be_2)}{x-a_{j,0}^*\pm \sin\theta t } \nonumber\\
= & \re\left( \cos \theta  \be_1 -  \sum_{j=1}^N \frac{2\ii s_{j,0}^*(\be_1\pm \ii\be_2)}{x-a_{j,0}^*\pm \sin\theta t } \right) + \sin\theta  \be_3
\end{align}
subject to \eqref{travelingwaveconstraint} at time $t=0$. To solve \eqref{travelingwaveconstraint} at $t=0$, we note that the Blaschke product in \eqref{blaschke} has the decomposition
\begin{equation}\label{blaschkedecomposition}
B(z)=1+\sum_{k=1}^N \frac{B_k}{z-a_{k,0}^*},\qquad B_k=(a_{k,0}-a_{k,0}^*)\pprod{j}{k}{N}  \frac{a_{j,0}-a_{k,0}^*}{a_{j,0}^*-a_{k,0}^*}.
\end{equation}
Because $B(a_{j,0})=0$, we must have
\begin{equation}
1+\sum_{k=1}^N \frac{B_k}{a_{j,0}-a_{k,0}^*}=0 \qquad (j=1,\ldots,N),
\end{equation}
and we see that  $2\ii s_{k,0}^*=-B_k \cos\theta$ provides a solution to \eqref{travelingwaveconstraint} at $t=0$. Thus, 
\begin{equation}
\cos\theta-\sum_{k=1}^N \frac{2\ii s_{j,0}^*}{x-a_{j,0}^*\pm \sin\theta t}=\cos\theta B(x \pm  \sin\theta  t). 
\end{equation}
Inserting this into \eqref{travelingwavesolution} gives
\begin{equation*}
\bm(x,t)=\be_1 \cos\theta\,  \re  B(x\pm \sin \theta t) \mp  \be_2 \cos\theta\, \im   B(x\pm \sin \theta t) + \be_3 \sin \theta,
\end{equation*}
which is \eqref{travelingwave} with $v=\mp \sin\theta$. 

\subsubsection{One-soliton solutions}
\label{sec:onesoliton} 
We write the traveling-wave solution in \eqref{travelingwave} in a coordinate-independent way by renaming $\be_1\to \bn_1$, $\mp\be_2\to \bn_2$, $\mp\be_3\to \bn_1\wedge\bn_2\coloneqq \bn$: 
\begin{subequations} 
\label{onesoliton} 
\begin{equation} 
\label{onesolitonbm} 
\bm(x,t) = (\bn_1\cdot\bmb)\bigl(\bn_1 \re B(x-vt)+\bn_2 \im B(x-vt) \bigr)+(\bn\cdot\bmb)\bn ,\quad v= \bn\cdot\bmb
\end{equation} 
with $B(z)$ in \eqref{blaschke} satisfying $|B(x-vt)|=1$. 
This makes manifest that, for fixed vacuum $\bmb\in S^2$, a traveling wave solution is determined by the direction $\bn\in S^2$ fixing its velocity, and this direction is at the same time a rotation axis: 
the solution describes a counter-clockwise rotation of $\bm$ about a circle in the plane orthogonal to $\bn$, starting out at $\bm\to \bmb$ at $x\to-\infty$ and ending again at $\bm\to \bmb$ at $x\to\infty$; 
$N\in\Z_{\geq 1}$ corresponds to the number of rotations of $\bm$ about $\bn$. 
Moreover, despite the fact that the propagation speed $v$ can be positive, zero, or negative, the traveling wave solutions are chiral in the sense that the 
rotation direction always is counter-clockwise when going from $x=-\infty$ to $x=+\infty$: $\bn_1\wedge\bn_2=+\bn$. 

In particular, the one-soliton solutions are given by \eqref{onesolitonbm} with 
\begin{equation} 
B(x-vt) =  \frac{x-\aR_{1,0}-\ii\aI_{1,0}- vt}{x-\aR_{1,0}+\ii\aI_{1,0}- vt} =  \frac{(x-\aR_{1,0}-vt)/\aI_{1,0}-\ii}{(x-\aR_{1,0}-vt)/\aI_{1,0}+\ii} \end{equation} 
\end{subequations} 
where $\aR_{1,0}\coloneqq\re a_{1,0}$ and $\aI_{1,0}\coloneqq \im a_{1,0}$; we interpret the traveling wave solutions for $N>1$ as a bound state of $N$ one-solitons which have the same constant rotation direction and speed. 

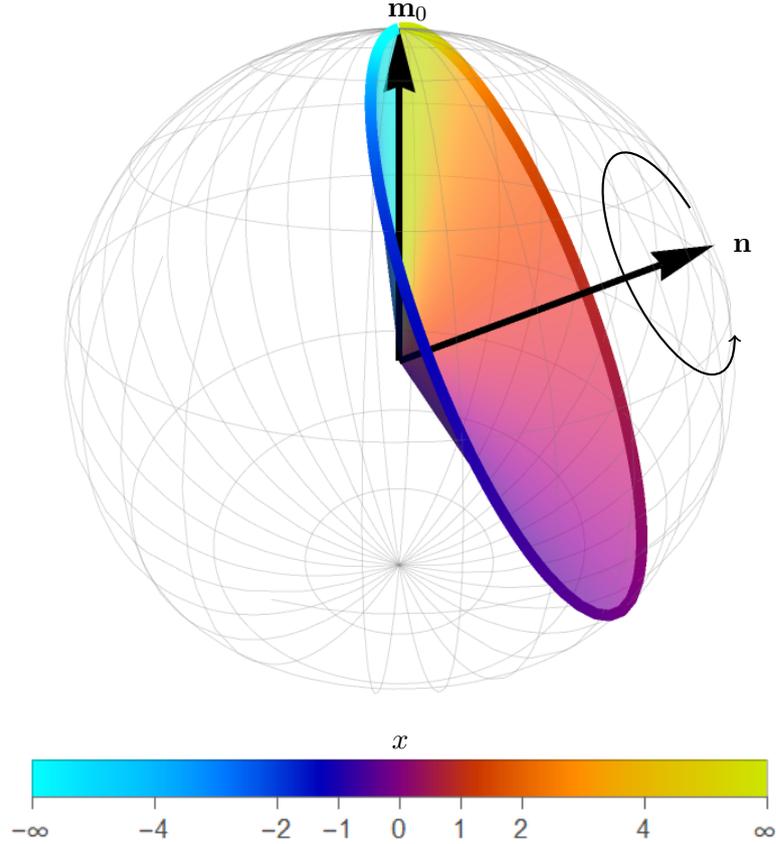
\begin{figure}[h!]
\centering
\begin{tikzpicture}
\def\a{1.6};
\def\b{20};
\def\c{0.6};
\node at (0,0) {\includegraphics[scale=0.7]{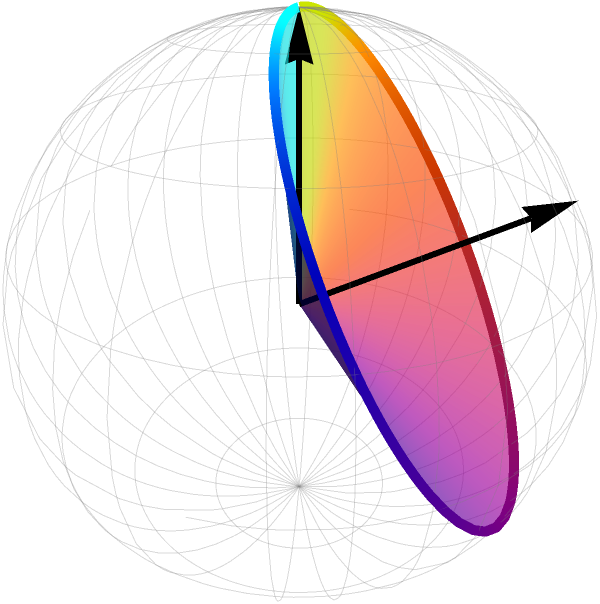}};
\node at (4.5,1.5) {$\mathbf{n}$};
\node at (0.1,4.6) {$\bmb$};
\node at (-0.26,-5.8) {\includegraphics[scale=1]{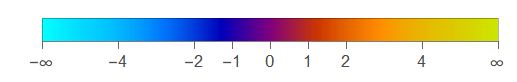}};
  \coordinate (P) at ($(0, 0) + (\b:\c cm and \a cm)$);
\node[above] at (0,-5.3) {$x$};
\draw[thick, black,->,rotate=25] (4.3,0.2) arc
  (\b:310:\c cm and \a cm);
\end{tikzpicture}

\caption{Spin configuration $\bm(x,t)$ for a one-soliton solution corresponding to the initial data \eqref{eq:initial_data_one_soliton} at fixed time $t=1/2$. 
Shown are the vacuum solution $\bmb$, the rotation direction $\bn$, and the spin $\bm(x,t)$ as a function of $x\in\R$, sweeping out a colored cone: at fixed $x$, $\bm(x,t)$ is drawn as a colored vector from the origin with its color indicating the value of $x$ according to the legend underneath. The curved arrow indicates the counter-clockwise rotation direction of the spin when going from $x=-\infty$ to $x=\infty$.}
\label{fig:single_soliton}
\end{figure}

\subsection{Solution of constraints}
\label{sec:constraints1}
We show that, for fixed vacuum $\bmb\in S^2$ and soliton number $N\in\Z_{\geq 1}$, Theorem~\ref{thm} provides a $4N$-parameter family of soliton solutions.
For that, we discuss the solution of the constraints in  \eqref{constraint}. For simplicity, we write in this section $a_j$ and $\bs_j$ short for $a_{j,0}$ and $\bs_{j,0}$, 
respectively, and we restrict our discussion to the case 
\begin{equation} 
\label{separated} 
|\aR_j/\aI_j-\aR_k/\aI_k|\gg 0\quad (1\leq j <k\leq N)  , 
\end{equation} 
where $\aR_j\coloneqq \re a_j$ and $\aI_j\coloneqq \im a_j$. We find it convenient to write \eqref{constraint} as
\begin{equation}\label{constraint11}
\bs_j^2=0, \quad \bs_j\cdot \bigg(\ii \bm_j -  \frac{\bs_j^*}{a_j-a_j^*} \bigg),\quad 
\bm_j \coloneqq \bmb+ \ii \ssum{k}{j}{N} \bigg( \frac{\bs_k}{a_j-a_k}- \frac{\bs_k^*}{a_j-a_k^*} \bigg) 
\end{equation}
for $j=1,\ldots,N$. 

As explained in Appendix~\ref{app:bs}, the general solution of the first constraint  in \eqref{constraint11} is 
\begin{equation} 
\label{bsjansatz}
\bs_j = s_j(\bn_{j,1}+\ii\bn_{j,2})
\end{equation} 
with $s_j\in\C$ and $\bn_{j,1},\bn_{j,2}\in S^2$ such that  $\bn_{j,1}\cdot \bn_{j,2}=0$, and these orthonormal vectors $\bn_{j,1}$ and $\bn_{j,2}$ are not unique but can be rotated in the plane spanned by them at the cost of changing the phase of $s_j$;  see Lemma~\ref{lem:sreps} (a) for a precise formulation of this fact. Thus, to solve the constraints \eqref{constraint11}, we pick $N$ poles $a_j\in\C_+$ and $N$ unit vectors $\bn_j\in S^2$ such that $\bn_j\cdot\bs_j=0$, i.e., $\bn_j=\bn_{j,1}\wedge\bn_{j,2}$. As argued below, for fixed $\bmb\in S^2$,  this choice can be made without any restriction, and then the constraints  \eqref{constraint11} determine the $\bs_j$ uniquely. 
Thus, solutions $(a_j,\bs_j)_{j=1}^N$ of \eqref{constraint} are parametrized by $4N$  real parameters: two for the real- and imaginary parts of each pole $a_j$; two for the polar angles determining the direction $\bn_j$; $j=1,\ldots,N$. 
Moreover, these parameters have a simple physical interpretation when the solitons are far apart: $\aR_j=\re a_j$ determines the initial position of the soliton $j$; $\aI_j=\im a_j$ determines its spatial extension; $\bn_j$ is the direction about which the spin rotates and which, at the same time, determines the soliton velocity through \eqref{onesolitonbm}.

To substantiate the claims in the previous paragraph, we insert \eqref{bsjansatz} in \eqref{constraint11}:
\begin{equation*} 
s_j(\bn_{j,1}+\ii\bn_{j,2})\cdot\bigg(\ii\bm_j - \frac{s_j^*(\bn_{j,1}-\ii\bn_{j,2})}{2\ii\aI_j} \biggr) =\ii s_j\bigg((\bn_{j,1}+\ii\bn_{j,2})\cdot\bm_j+ \frac{s_j^*}{\aI_j} \bigg) =0, 
\end{equation*} 
which has the solution 
\begin{equation} 
\label{bsjsolution}
s_j = \aI_j (\bn_{j,1}-\ii\bn_{j,2})\cdot\bm^*_j\quad (j=1,\ldots,N). 
\end{equation} 
Since $\bm_j$ depends on $\bs_{k\neq j}$, \eqref{bsjsolution} together with the definition of $\bm_j$ in \eqref{constraint11} provide a system of $N$ non-linear equations for $N$ unknowns $s_j$. 
We found that this system of equations can be solved efficiently by an iterative procedure if the conditions in \eqref{separated} hold true, which can be understood as follows:  
at initial time $t=0$, one has $N$ well-separated one-solitons, and  $\bm_j^{(0)}=\bmb$ independent of $j$ is a reasonable approximation to $\bm_j$. 
Inserting this in \eqref{bsjsolution}, one obtains an approximation to $\bs_j$: $\bs_j^{(1)}= -\aI_j (\bn_{j,1}-\ii\bn_{j,2})\cdot\bmb$, which can be inserted in the definition of $\bm_j$ to get a better approximation $\bm_j^{(1)}$ to $\bm_j$, etc. 
This suggests the following iteration scheme to solve this system of equations: 
\begin{equation}
\label{eq:iteration_idea}
s_j^{(n+1)} = \aI_j (\bn_{j,1}-\ii\bn_{j,2})\cdot\biggl( \bmb- \ii \ssum{k}{j}{N} \bigg( \frac{(s^{(n)}_k)^* (\bn_{j,1}-\ii\bn_{j,2})}{a_j^*-a_k^*}- \frac{s^{(n)}_k(\bn_{j,1}+\ii\bn_{j,2})}{a_j^*-a_k} \bigg)  \biggr) 
\end{equation} 
 for $n=0,1,2,\ldots$, starting out with $s_j^{(0)}\coloneqq 0$.
As explained in Section~\ref{sec:constraints}, we have implemented this iterative scheme, and we found that it converges well to a unique solution $\bs_j$ for a majority of choices for $a_j$ and $\bn_j$ even for small soliton separations. 
Moreover, we found that  the $N$-soliton solution is well-approximated by a superposition of one-soliton solutions: 
\begin{subequations} 
\label{largetimes}
\begin{equation} 
\bm(x,t)\approx  \sum_{j=1}^N (\bn_{1,j}\cdot\bmb)\bigl(\bn_{1,j}\re B_j(x-v_j t) + \bn_{2,j}\im B_j(x-v_j t)\bigr) +(\bn_j\cdot\bmb)\bn_j , 
\end{equation} 
\begin{equation} 
\bn_j=\bn_{j,1}\wedge\bn_{j,2}, \quad B_j(x-v_j t) = \frac{(x-\aR_j-v_jt)/\aI_j-\ii}{(x-\aR_j-v_jt)/\aI_j+\ii},\quad  v_j= \bn_j\cdot\bmb,
\end{equation} 
provided 
\begin{equation} 
|(\aR_j-v_jt)/\aI_j - (\aR_k-v_k t)/\aI_k|\gg 1\quad (1\leq j<k\leq N)
\end{equation} 
\end{subequations} 
for all times $-t_0\leq t\leq t_0$ and some $t_0>0$, and this makes manifest the physical interpretation of the parameters $a_j\in\C$ and $\bn_j$ given above (to obtain this, we make use of \eqref{U1} to rotate $\bn_{j,1},\bn_{j,2}$ so that $s_j>0$). 
Moreover, at other time intervals where the solitons are well-separated, the solution is again well approximated by a formula like \eqref{largetimes} but with possible phase shifts $\aR_j\to \taR_j$ caused by soliton-soliton interactions at intermediate times. 
In particular, any multi-soliton solution is well approximated by \eqref{largetimes} in the far past ($t\to -\infty$) and the far future ($t\to\infty$), with possible phase shifts in $\aR_j\to \aR_{j,\mp}$.
As an illustration of this, we plot a two- and a three-soliton solution in the far past and future in Fig. \ref{fig:soliton_asymp},  
clearly showing that in those regimes the multi-soliton solution is well-approximated by a superposition of one-soliton solutions.

\begin{figure}[h!]
\centering
\begin{tikzpicture}
\def\a{7};
\def\b{3.7};
\def\c{4.25};
\node at (0,0) {\includegraphics[scale=1]{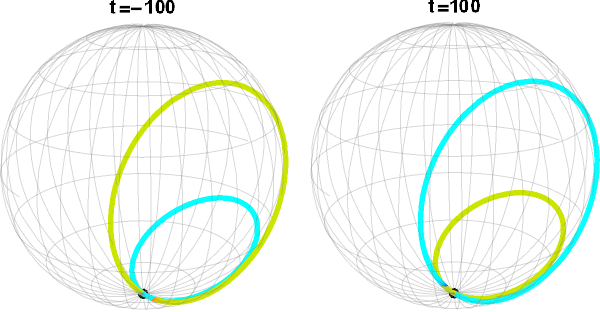}};
\draw[gray] (-\a,\b) -- (\a,\b);
\draw[gray] (-\a,-\b) -- (\a,-\b);
\foreach \x in {-1,0,1}
{
\draw[gray] (\x*\a,\b) -- (\x*\a,-\b);
}
\fill[white] (-4,2.9) rectangle (-2,3.6);
\fill[white] (2,2.9) rectangle (6,3.6);
\node at (-3.3,3.2) {$t=-100$};
\node at (3.3,3.2) {$t=100$};
\end{tikzpicture} \\
\vspace{-1pt}
\begin{tikzpicture}
\def\a{7};
\def\b{3.7};
\def\c{4.25};
\node at (0,0) {\includegraphics[scale=1]{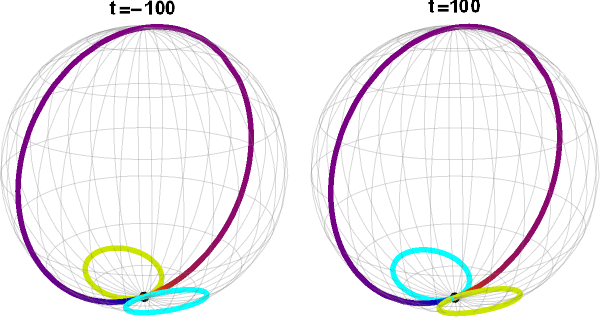}};
\draw[gray] (-\a,\b) -- (\a,\b);
\draw[gray] (-\a,-\b) -- (\a,-\b);
\foreach \x in {-1,0,1}
{
\draw[gray] (\x*\a,\b) -- (\x*\a,-\b);
}
\fill[white] (-4,2.9) rectangle (-2,3.5);
\fill[white] (2,2.9) rectangle (6,3.5);
\node at (-3.3,3.2) {$t=-100$};
\node at (3.3,3.2) {$t=100$};

\node at (-0.26,-\b-1.3) {\includegraphics[scale=1]{legend.png}};
\node[above] at (0,-\b-0.8) {$x$};
\end{tikzpicture}

\caption{Spin configurations $\bm(x,t)$ for the two- and three-soliton solutions with initial data \eqref{eq:initial_data_two_soliton} and \eqref{eq:initial_data_three_soliton}, respectively, at fixed times in the far past and future. 
The curves on the unit sphere show the spatial dependence of the spin $\bm(x,t)$, with the color indicating the value of $x\in\R$ according to the legend underneath; the black dot indicates the vacuum direction $\bmb$. 
We clearly see the individual solitons: each such soliton corresponds to a circle starting and ending at the vacuum direction $\bmb$. The time evolution of these two- and three-soliton solutions is given in Fig.~\ref{fig:two_soliton} and \ref{fig:three_soliton}, respectively.}
\label{fig:soliton_asymp}
\end{figure}

\section{Periodic multi-soliton solutons}
\label{sec3}

The periodic HWM equation is defined by the same equation \eqref{hwm} as in the real-line case, but the $S^2$-valued functions $\bm(x,t)$ are periodic: $\bm(x,t)=\bm(x+L,t)$, with $L>0$ a fixed parameter, 
and instead of the Hilbert transform in \eqref{hilbert_line} one uses its periodic generalization \cite{king2009}: 
\begin{equation} 
\label{hilbert_circle}
(Hf)(x)\coloneqq\frac{1}{L}\pvint_{-L/2}^{L/2}\cot\left(\frac{\pi}{L}(x'-x)\right)f(x')\, \mathrm{d}x'  .
\end{equation} 
As proved in Appendix~\ref{app:proofs}, all results discussed in Section~\ref{sec2} can be straightforwardly generalized to this periodic case by replacing $\alpha(x)$ and $V(x)$ in \eqref{alphaV} with 
\begin{equation} 
\label{alphaV1} 
\alpha(x)  \coloneqq \kappa\cot\kappa x  ,\quad V(x) \coloneqq \frac{\kappa^2}{\sin^{2}\kappa x}  , \quad \kappa  \coloneqq\frac{\pi}{L}   ,
\end{equation} 
where $V(x)$ is the interaction potential of the trigonometric CM system and $\alpha(x)$ the associated special function such that $V(x)=-\alpha'(x)=\alpha(x)^2-\kappa^2$  \cite{calogero1975}. 
In particular, the pole ansatz that gives multi-soliton equations for the periodic HWM equation is 
\begin{equation}\label{ansatz1}
\bm(x,t)=\bmb+\ii\sum_{j=1}^N \bs_j(t) \kappa\cot\kappa\left(x-a_j(t)\right)-\ii\sum_{j=1}^N \bs_j(t)^* \kappa\cot\kappa\left(x-a_j(t)^*\right)   . 
\end{equation}
 
\begin{theorem} 
\label{thm1} 
For arbitrary $N\in\Z_{\geq 1}$, $\bmb\in\R^3$, $\bs_{j,0}\in\C^3$ and $a_{j,0}\in\C_+$  such that 
\begin{subequations} 
\label{constraintthm1} 
\begin{equation} \label{constraintthm1a} 
\bs_{j,0}^2=0,
\quad \bs_{j,0}\cdot\bigg(\ii\bmb - \ssum{k}{j}{N}\bs_{k,0}\kappa\cot\kappa\left(a_{j,0}-a_{j,0}\right)+  \sum_{k=1}^N\bs^*_{k,0}\kappa\cot\kappa\left(a_{j,0}-a_{j,0}^*\right)\bigg)=0 
\end{equation} 
for $j=1,2,\ldots,N$,  and 
\begin{equation}\label{newconstraint} 
\bmb^2- 4\kappa^2\bigg( \sum_{j=1}^{N}   \im  \bs_{j,0} \bigg)^2  = 1, 
\end{equation}
\end{subequations} 
let $\bs_j(t)$ and $a_j(t)$ be solutions of the following system of equations,  
\begin{subequations} 
\label{dyn1}
\begin{align} 
\dot\bs_j(t) 
& = -2 \ssum{k}{j}{N}\bs_j(t)\wedge\bs_k(t)\frac{\kappa^2}{\sin^{2}\kappa \left( a_j(t)-a_k(t) \right)}\label{sjdotthm1}\\
\ddot a_j(t) & = 4\ssum{k}{j}{N}\bs_j(t)\cdot\bs_k(t)\frac{\kappa^3\cos\kappa\left(a_j(t)-a_k(t)\right)}{\sin^3\kappa\left(a_j(t)-a_k(t)\right)} \label{ajddotthm1},
\end{align} 
\end{subequations} 
with initial conditions $\bs_j(0)=\bs_{j,0}$, $a_j(0)=a_{j,0}$, and 
\begin{align} 
\dot a_j(0) = & \frac{\bs_{j,0}\wedge \bs_{j,0}^*}{\bs_{j,0}\cdot\bs_{j,0}^*}\cdot\bigg( \ii\bmb   - \ssum{k}{j}{N} \bs_{k,0}\kappa\cot\kappa\left(a_{j,0}-a_{k,0}\right) 
+ \sum_{k=1}^N \bs^*_{k,0}\kappa\cot\kappa\left(a_{j,0}-a^*_{k,0}\right)\bigg)
\end{align} 
for $j=1,2,\ldots,N$. Then, for all times $t\in \R$ such that (i) $a_j(t)\in\C_+$ for all $j=1,\ldots,N$, (ii) $\bm(x,t)$ is differentiable with respect to $x$ and $t$, $\bm(x,t)$ in \eqref{ansatz1} is an exact solution of the periodic HWM equation \eqref{hwm} and \eqref{hilbert_circle} such that $\bm(x,t)^2=1$. 
\end{theorem} 

(The proof is given in Appendix~\ref{app:proofs}.) 

Clearly, this result is a straightforward generalization of the one in Theorem~\ref{thm} for the rational case; the only new feature is the appearance of an additional constraint on initial conditions in \eqref{newconstraint}. 

Similarly to the rational case, the spin-pole dynamics in \eqref{dyn1} is identical to the equations of motion of the trigonometric spin CM model; see Appendix~\ref{app:EOM} for details. 

\section{Generating solutions}
\label{sec:constraints}

We have thus far shown that our soliton ansatz  \eqref{ansatz} provides a solution to the real-line HWM equation if the initial data satisfy the constraints \eqref{constraint}, and we presented a similar result in the periodic case with the  corresponding constraints in \eqref{constraintthm1a}--\eqref{newconstraint}. 
In this section we explore several ways to solve these constraints,  and we provide explicit examples of multi-soliton solutions.

\subsection{Iteration algorithm}
\label{sec:iteration} 
As already discussed in Section \ref{sec:constraints1} real-line case, one can find solutions to the constraints  \eqref{constraint} satisfied by soliton initial conditions using iteration. We discuss the algorithm for the corresponding constraints \eqref{constraintthm1a}--\eqref{newconstraint}  in the periodic case. This requires the incorporation of one extra constraint compared to the real-line case; the restriction to the real-line case is straightforward, as will be explained.

If we ignore soliton interactions,  the constraints  \eqref{constraintthm1a}--\eqref{newconstraint}  simplify to 
\begin{equation}
\label{eq:algorithm_basic}
\mathbf{s}_j^2 = 0  , \quad \mathbf{s}_j\cdot \bigg(\bmb-\ii \mathbf{s}_j^*  \alpha(a_j-a_j^*) \bigg)=0  , \quad \bmb^2-4\kappa^2 \sum_{j=1}^{N}{\sum_{k=1}^N}  \im  \bs_j \cdot \im  \bs_k = 1, 
\end{equation}
with $\alpha(x)$ in \eqref{alphaV1}. 
We fix some poles $a_j \in \C_+$ and a unit vector $\mathbf{n}_0 \in S^2$, write $\bmb = m  \mathbf{n}_0$ for some $m>0$, and employ again the decomposition suggested by Lemma~\ref{lem:sreps}: 
\begin{equation}
\label{eq:decomposition}
\mathbf{s}_j = s_j ( \bn_{j,1} +  \ii \bn_{j,2} )  ,
\end{equation}
with two unit vectors $\bn_{j,1},\bn_{j,2} \in S^2$. By requiring $s_j>0$ we fix the $\mathrm{U}(1)$ freedom discussed in Section \ref{sec:constraints1} and Lemma \ref{lem:sreps}, and we get a unique representation of $\bs_j$:  we find from \eqref{eq:algorithm_basic} 
\begin{equation}
\label{eq:nRI_constraint}
\bn_{j,1}\cdot \bn_{j,2} = 0  , \qquad \bn_{j,2}\cdot \bmb = 0  , 
\end{equation}
and choosing such unit vectors $\bn_{j,1}$ and $\bn_{j,2}$, we find from \eqref{eq:algorithm_basic} that $s_j   = m s_j'$ with 

\begin{equation}
\label{eq:iteration_base_step}
s_j' = \frac{\bn_{j,1}\cdot \mathbf{n}_0}{2\ii  \alpha(a_j - a_j^*)},
\end{equation}
with the constant $m>0$ determined by the last constraint in \eqref{eq:algorithm_basic}: 

\begin{equation}
\label{eq:iteration_base_step_m}
m^2 \left( 1- 4 \kappa^2 \sum_{j=1}^N \sum_{k=1}^N s_j' s_k'\bn_{j,2} \cdot \bn_{k,2} \right) =1. 
\end{equation}

This solves the simplified constraints  \eqref{eq:algorithm_basic}. 

The full constraints are as in \eqref{eq:algorithm_basic} but with $\bmb$ in the second equation replaced by 
\begin{equation} 
\bm_j = \bmb + \ii  \ssum{k}{j}{N}\bs_{k}\alpha\left(a_{j}-a_{j}\right) -\ii  \sum_{k=1}^N\bs^*_{k}\alpha\left(a_{j}-a_{j}^*\right). 
\end{equation} 
To solve this, we take the unit vectors $\mathbf{n}_0$, $\bn_{j,1}$ and $\bn_{j,2}$ as above, and, similarly as above, the parameters $s_j$ and $m$ are determined by 
\begin{equation}
\label{eq:iteration}
s_j   =\frac{(\bn_{j,1} - \ii \bn_{j,2})\cdot \mathbf{m}_j^*}{2\ii  \alpha(a_j - a_j^*)}  , \quad  m^2 = 1 +4\kappa^2 \sum_{j=1}^{N}{\sum_{k=1}^N}  \im  \bs_j \cdot \im  \bs_k.
\end{equation}
We solve the latter non-linear system of equations by iteration, taking the approximate solution in  \eqref{eq:iteration_base_step}--\eqref{eq:iteration_base_step_m} as base step: $s_j^{(0)}=m s_j'$ and $m^{(0)}=m$. 
We then compute a sequence $\{s_j^{(n+1)}\}$ and $m^{(n+1)}$ for $n=0,1,2,\ldots$: given a set of $\{s_j^{(n)}\}$ and $m^{(n)}$,  we define the associated $\bm_j^{(n)}$ and use \eqref{eq:iteration} to compute a new set $\{s_j^{(n+1)}\}$ as well as a new $m^{(n+1)}$. Since the right-hand side of the first equation in \eqref{eq:iteration} can be non-real during the iteration, we allow the $s_j$ to become complex. As shown in Lemma \ref{lem:sreps}, this extension effectively allows for the rotation of the vectors $\bn_{j,1},\bn_{j,2}$ in the plane spanned by them while keeping them orthogonal; however, the unit vectors $\bn_j=\bn_{j,1}\wedge\bn_{j,2}$ remain unchanged. 

The discussion above can be summarized in the following algorithm: 
\begin{enumerate}
\item Fix a set of poles $a_j$ and a unit vector $\bn_0$. 
\item For $j=1,\ldots,N$,  choose a unit vector $\bn_{j,2}$ orthogonal to $\bn_0$  and then a unit vector $\bn_{j,1}$ orthogonal to $\bn_{j,2}$. 
\item Compute the set $s_j^{(0)}=ms_j'$, $j=1,\ldots,N$,  and $m^{(0)}=m$ using \eqref{eq:iteration_base_step} and \eqref{eq:iteration_base_step_m}. 
\item Find new sets $\{s_j^{(n+1)}\}$ and $m^{(n+1)}$ iteratively, evaluating the right-hand side of \eqref{eq:iteration} using the data acquired at step $n$. 
\end{enumerate}
We repeat step $4$ until the result stabilizes to some given accuracy. 

In order to apply the algorithm to the rational case, we note that the constraint \eqref{eq:iteration_base_step_m} simply becomes $m^2 =1 $ in that case. 
This means $m$ does not change during the iteration and we can keep it fixed to $m=1$.

\subsubsection{Results}
We implemented the algorithm in Mathematica\footnote{Notebooks attached to the arXiv submission of this paper contain this algorithm and the data discussed here. In particular, it contains the data used to generate Fig. \ref{fig:iteration_data}.} and used random input data to find sets $\{a_j,\bs_j\}$ that give soliton initial conditions to the HWM equation. 
As long as the chosen poles $a_j$ were not too close together to cause numerical instabilities, the algorithm converges and yields soliton initial conditions that satisfy the constraints with our requested accuracy of 10 digits. 
More specifically, we found that, for $N=2,3,4,5$, for 10000 randomly chosen initial directions $\bn_{j,1}$ and $\bn_{j,2}$ and poles $a_j$, $\sim 90\%$ of the cases that satisfied $|\aR_j/\aI_j-\aR_k/\aI_k| > 1$ for all $j\neq k$ had converged to $1$ and $>$10 digits after 40 and 150 iterations, respectively.   
Examples of soliton initial conditions found in this way are given in Fig. \ref{fig:iteration_data}. 
Admissible initial conditions can be time-evolved using \eqref{dyn} to obtain a full solution to the HWM equation. Before discussing this, let us first explain how one in special cases can even find exact solutions of the constraints \eqref{constraintthm} and \eqref{constraintthm1a}.

\begin{figure}[h!]
\begin{center}
\begin{tikzpicture}
\def\a{4.};
\def\b{4.};
\def\c{2};

\node at (-0.56*\a,\b+3) {\includegraphics[scale=1]{legend.png}};
\node[above] at (-0.56*\a+0.26,\b+3.5) {$x$};
\node at (-2*\a,\b) {\includegraphics[width=\textwidth/5]{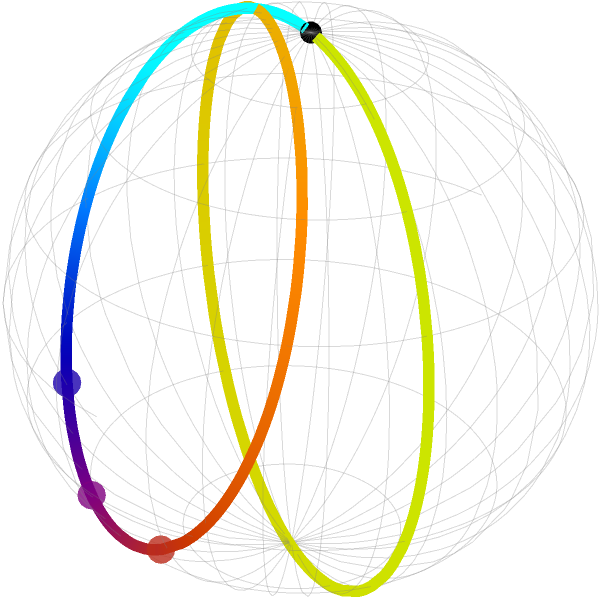}};
\node at (-\a,\b) {\includegraphics[width=\textwidth/5]{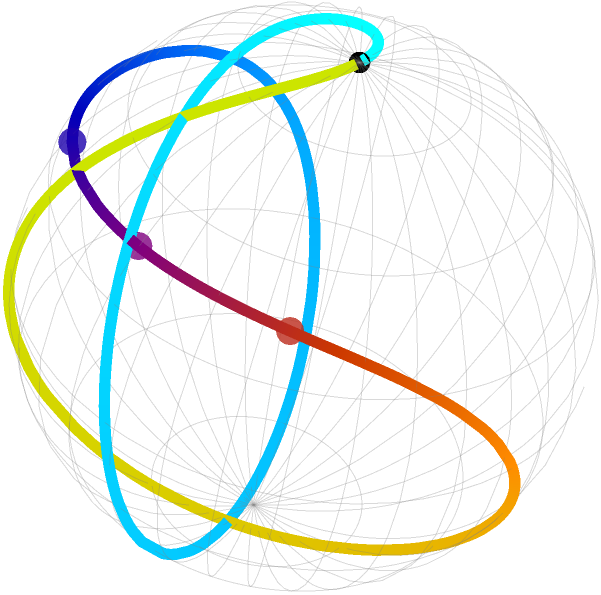}};
\node at (0,\b) {\includegraphics[width=\textwidth/5]{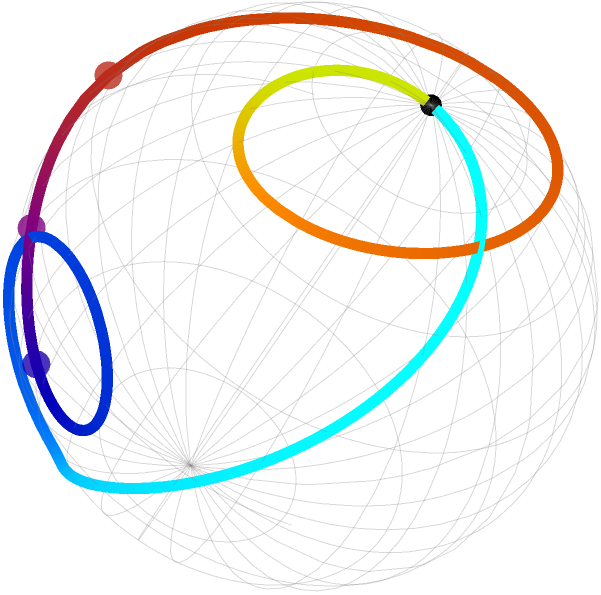}};
\node at (\a,\b) {\includegraphics[width=\textwidth/5]{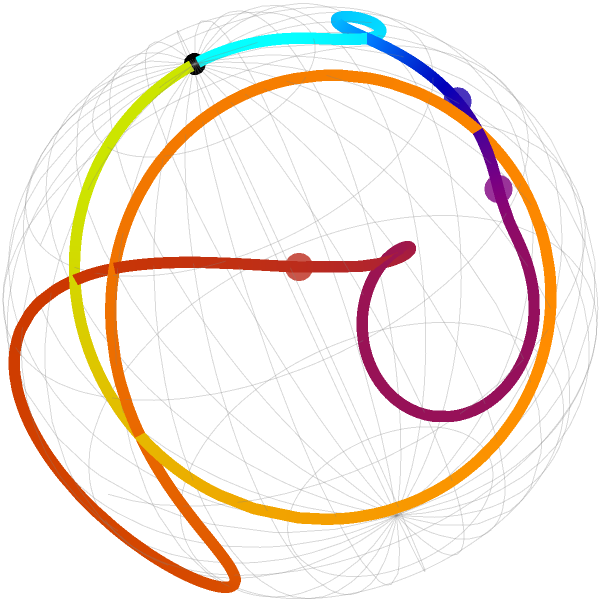}};

\node at (-2*\a,0) {\includegraphics[width=\textwidth/5]{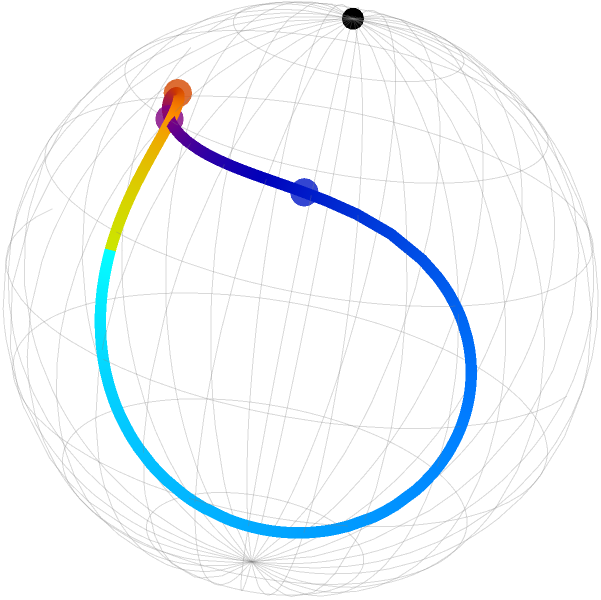}};
\node at (-\a,0) {\includegraphics[width=\textwidth/5]{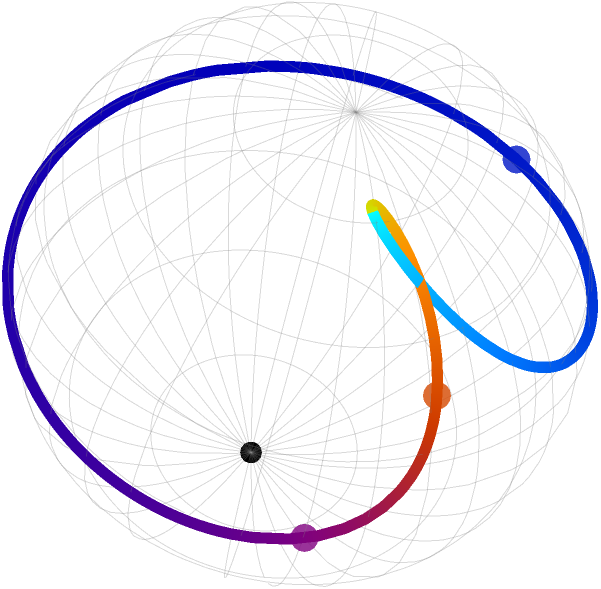}};
\node at (0,0) {\includegraphics[width=\textwidth/5]{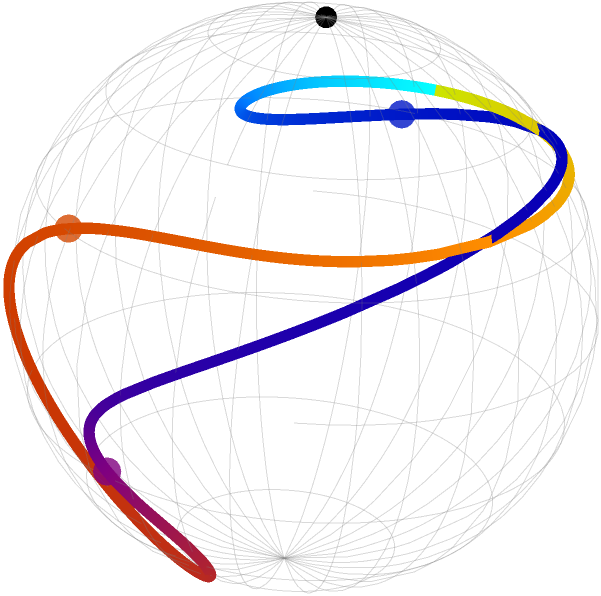}};
\node at (\a,0) {\includegraphics[width=\textwidth/5]{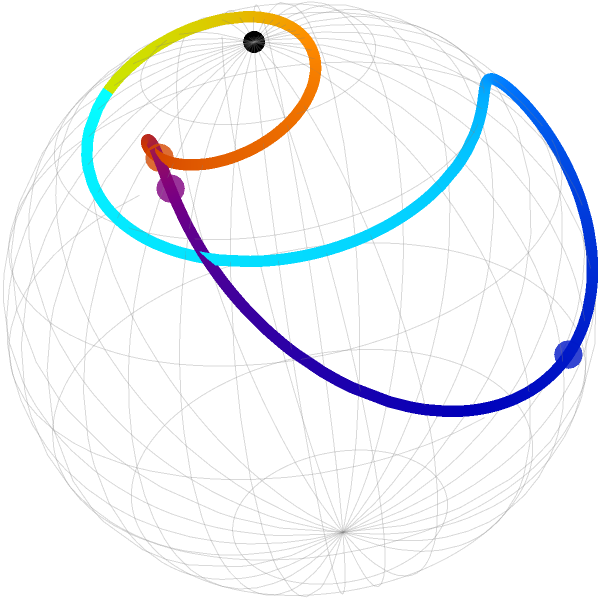}};

\node at (-0.56*\a,-3.3) {\includegraphics[scale=1]{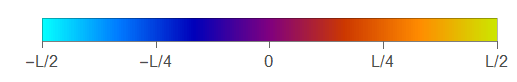}};
\node at (-0.56*\a+0.26,-2.5) {$x$};

\foreach \x in {-1,0,1}
{
\draw[gray] (-2*\a-\c,-\x*\b+0.5*\b) -- (\a+\c,-\x*\b+0.5*\b);
};

\foreach \x in {0,...,4}
{
\draw[gray] (-2*\a+\x*\a-\c,1.5*\b) -- (-2*\a+\x*\a-\c,-0.5*\b);
};
\end{tikzpicture}
\end{center}
\caption{Examples of multi-soliton initial conditions $\bm(x,t=0)$ for the HWM equation. 
The upper and lower four plots correspond to the real-line and periodic cases, respectively, and the soliton numbers are $N=2,3,4,5$ from left to right.
The position $x$ is given by colors corresponding to the legends above and below in the real-line and periodic cases, respectively; $L=5$ in the latter case. 
The vacuum $\bmb$ is indicated by a black dot, and the colored dots correspond to $\bm(x,t=0)$ at $x=-1$ (blue), $x=0$ (purple) and $x=1$ (red), respectively. 
For fixed $N$, we generated random poles $a_j$ and random directions $\bn_{j,1}$ and $\bn_{j,2}$, and we solved the non-linear constraints determining soliton initial conditions using an iterative method, as explained  in Section~\ref{sec:iteration}; the given examples illustrate the effect of inter-soliton interactions: as explained in Section~\ref{sec:constraints1}, in the real-line case and for times $t\to\pm \infty$, all spin configurations evolve towards $N$ distinct circles corresponding to non-interacting solitons.  
}
\label{fig:iteration_data}
\end{figure}

\subsection{Exact solutions}
\label{sec:exact} 
We describe a useful 
 heuristic method to generate analytic solutions to the constraints and thus obtain soliton initial conditions \eqref{ansatz} for \eqref{hwm}.

 A standard parameterization of $S^2$ is
\begin{equation}
\left( \sin\theta \cos \phi,   \sin\theta \sin\phi,   \cos\theta\right) 
\end{equation}
where $\theta\in[0,\pi]$ and $\phi\in[0,2\pi)$. Setting $\theta=\arctan y$ and $\phi=\arctan x$, we obtain
\begin{equation}\label{S2paramrational}
\left( \frac{2y(1-x^2)}{(1+y^2)(1+x^2)},    \frac{4xy}{(1+y^2)(1+x^2)},  \frac{1-y^2}{1+y^2} \right).
\end{equation}
By specifying a relationship between $y$ and $x$, we obtain a rational function from $\mathbb{R}$ to $S^2$. Doing this, we obtain three examples of soliton initial data.

\begin{enumerate}
\item $y=1/\sqrt{3}$. The corresponding initial data is
\begin{equation}
    \bm(x,0)=\bigg(\frac{2x}{1+x^2},0,\frac12\bigg).
    \end{equation}
    In terms of the expansion \eqref{ansatz} at $t=0$, the parameters are
    \begin{equation}\label{eq:initial_data_one_soliton}
       \bmb=(-\sqrt{3}/2,0,1/2),\quad a_{1,0}=\ii, \quad \bs_{1,0}=(\sqrt{3}/2,-\sqrt{3}\ii/2,0).
    \end{equation} 
    \item $y=x/2$. The corresponding initial data is
    \begin{equation}
    \bm(x,0)=\bigg(\frac{\sqrt{3}}{2}\frac{1-x^2}{1+x^2},\frac{\sqrt{3} x}{1+x^2},\frac12\bigg).
    \end{equation}
    In terms of the expansion \eqref{ansatz} at $t=0$, the parameters are
 \begin{equation}
 \begin{split}
\label{eq:initial_data_two_soliton}
\bmb =& (0,0,-1),\quad   (a_{1,0}, a_{2,0}) = (\ii, 2\ii), \\
 \mathbf{s}_{1,0} =& (4\ii/3,-4/3,0),\quad  \mathbf{s}_{2,0} = \left(-10\ii/3,8/3,2 \right). 
 \end{split}
 \end{equation}

    \item $y=x^2$. The corresponding initial data is
    \begin{equation}
    \bm(x,0)=\bigg(\frac{2x^2(1-x^2)}{(1+x^2)(1+x^4)},\frac{4x^3}{(1+x^2)(1+x^4)}    ,\frac{1-x^4}{1+x^4}\bigg).
    \end{equation}
    In terms of the expansion \eqref{ansatz} at $t=0$, the parameters are
    \begin{equation}
\label{eq:initial_data_three_soliton}
\begin{aligned}
\bmb &= (0,0,-1), \quad (a_{1,0}, a_{2,0}, a_{3,0}) = (\ii,e^{\ii\pi/4}, e^{3\ii\pi/4}),  \\
\mathbf{s}_{1,0} &= (-1,\ii,0), \quad \mathbf{s}_2 = \tfrac{1}{2}\left(e^{3\ii\pi/4} ,-1-\ii, e^{3\ii\pi/4} \right), \\
\mathbf{s}_{3,0} &= \tfrac{1}{2} \left(e^{5\ii\pi/4}, 1-\ii, e^{5\ii\pi/4} \right).
\end{aligned}
\end{equation}
\end{enumerate}

A similar method can be used to construct trigonometric initial data. Consider a meromorphic, $L$-periodic function ${\mathbf f}:\mathbb{C}\to \mathbb{C}^3$ satisfying ${\mathbf f}(z^*)={\mathbf f}(z)^*$. If ${\mathbf f}$ is bounded on $\{z\in \mathbb{C}: |\im z|>M\}$ for some positive number $M$, then Liouville's theorem ensures that ${\mathbf f}$ can be put into the form \eqref{ansatz1}. We can obtain such functions also satisfying ${\mathbf f}(\R)\subset S^2$ by setting $y$ and $x$ to appropriate trigonometric functions of the same variable in \eqref{S2paramrational}. However, there are complications in controlling the number of poles (which corresponds to the soliton number).

While the numerical iteration method explained in Section~\ref{sec:iteration} is far more flexible, the method in this section is interesting since it provides initial conditions leading to cusps, as discussed at the end of Section~\ref{sec:numerics}.

\subsection{Numerical time-evolution}
\label{sec:numerics}
As a check of our results as well as an aid to develop intuition,  we have developed two independent numerical implementations to solve \eqref{hwm} using the ansatz \eqref{ansatz}.

Examples of the time evolutions of two- and three-soliton solutions for the real-line case are displayed in Fig. \ref{fig:two_soliton} and Fig. \ref{fig:three_soliton}, respectively.\footnote{A Mathematica notebook to generate a video of the time evolution is available in the source files of the arXiv submission.} The figures also give the energy density\footnote{The sign difference in \eqref{eq:energy_density} versus \cite{zhou2015} is due to our convention for the Hilbert transform in \eqref{hilbert_line}.} \cite{zhou2015}
\begin{equation}
\label{eq:energy_density}
 \varepsilon \coloneqq - \bm \cdot H\bm_x
\end{equation} 
of the respective solutions. The energy density exhibits the individual solitons as localized peaks.
\begin{figure}[h!]
\begin{tikzpicture}
\def\a{-3.5};
\def\b{2.2};
\def\c{4.25};
\node at (0,0) {\includegraphics[width=\textwidth]{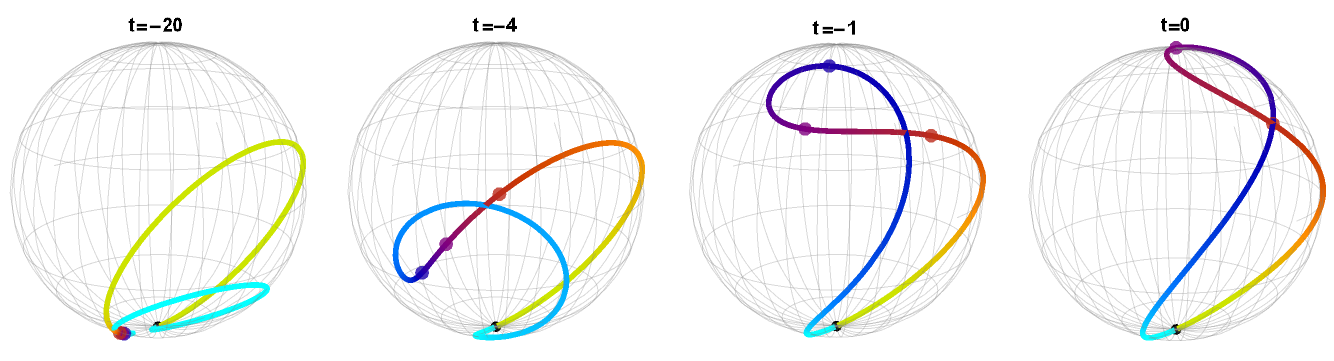}};
\node at (0,\a) {\includegraphics[width=\textwidth]{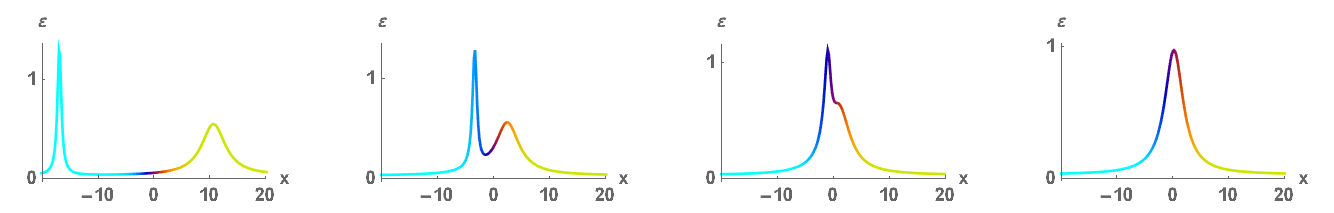}};
\node at (0,2*\a) {\includegraphics[width=\textwidth]{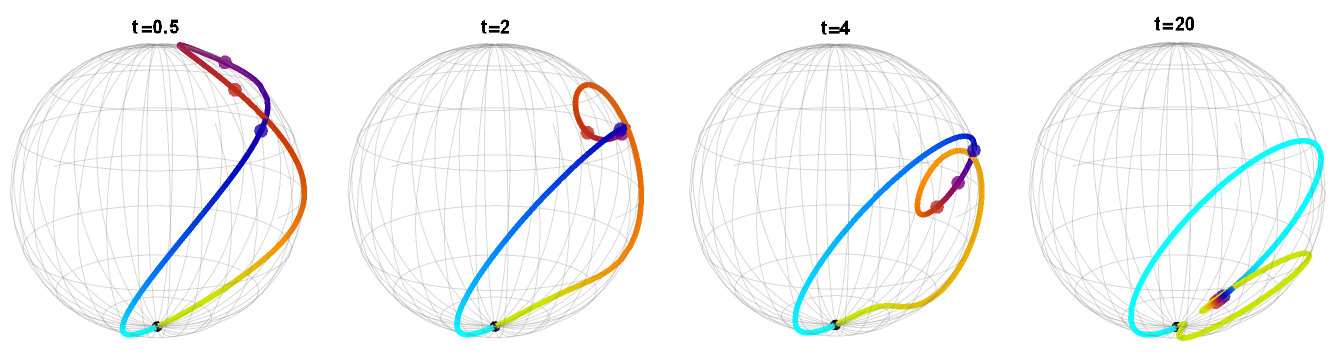}};
\node at (0,3*\a) {\includegraphics[width=\textwidth]{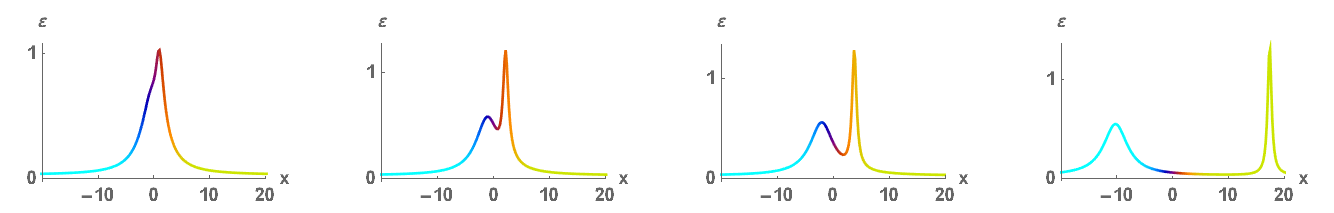}};
\node at (-0.26,3.7*\a-0.3) {\includegraphics[scale=1]{legend.png}};
\node at (0,3.7*\a+0.4) {$x$};
\foreach \x in {0,2,4}
{
\draw[gray] (-8.5,\b+\x*\a) -- (8.5,\b+\x*\a);
};

\foreach \x in {0,...,4}
{
\draw[gray] (-8.5+\x*\c,\b) -- (-8.5+\x*\c,\b+4*\a);
};

\foreach \y in {0,2}
{
\foreach \x in {0,...,3}
{
\fill[white] (-7.5+\x*\c,\b+\y*\a-0.5) rectangle (-4.5+\x*\c,\b+\y*\a-0.05);
}
};

\node at (-6.2,\b-0.3) {\scriptsize $t=-20$};
\node at (-6.4+\c,\b-0.3) {\scriptsize $t=-4$};
\node at (-6.4+2*\c,\b-0.3) {\scriptsize $t=-1$};
\node at (-6.4+3*\c,\b-0.3) {\scriptsize $t=0$};

\node at (-6.2,\b+2*\a-0.3) {\scriptsize $t=0.5$};
\node at (-6.3+\c,\b+2*\a-0.3) {\scriptsize $t=2$};
\node at (-6.4+2*\c,\b+2*\a-0.3) {\scriptsize $t=4$};
\node at (-6.4+3*\c,\b+2*\a-0.3) {\scriptsize $t=20$};

\end{tikzpicture}
\caption{The time evolution of the two-soliton solution with initial data \eqref{eq:initial_data_two_soliton}. Each frame shows the spin configuration $\bm(x,t)$ (colored line on unit sphere) and energy density in \eqref{eq:energy_density} (lower plot)  at one moment in time $t$ given at the top. The curve on the sphere shows the change of the spin $\bm(x,t)$ along the $x$-axis, with colors indicating the $x$-axis location according to the legend underneath. Note that the color gradient is non-linear, changing linearly only close to the origin $x=0$ but approaching a constant color away from the origin exponentially fast. The vacuum direction is indicated by a black dot (south pole). For three distinguished points $x=-1$ (blue), $x=0$ (purple) and $x=1$ (red) on the $x$-axis, we have indicated the corresponding spin $\bm(x,t)$ with a dot. The energy density plots show the solitons as localized lumps. 
}
\label{fig:two_soliton}
\end{figure}

\begin{figure}[h!]
\begin{tikzpicture}
\def\a{-3.5};
\def\b{2.2};
\def\c{4.25};
\node at (0,0) {\includegraphics[width=\textwidth]{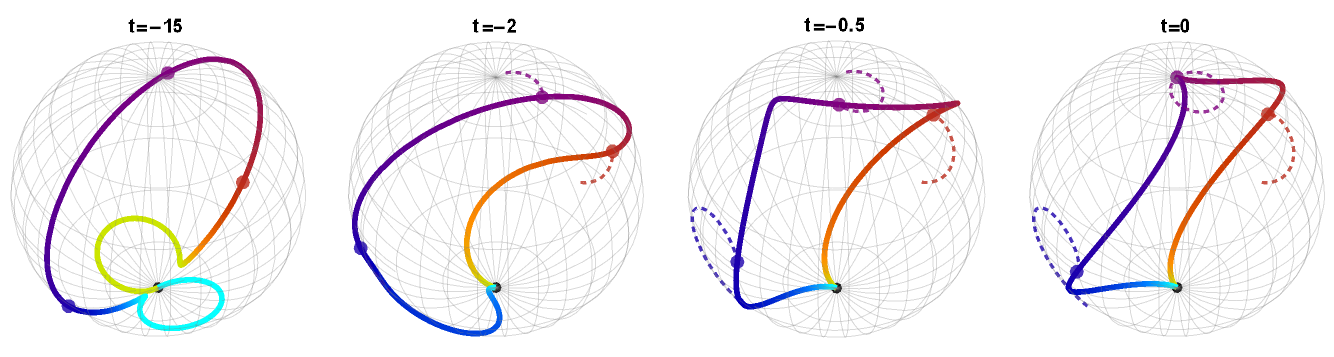}};
\node at (0,\a) {\includegraphics[width=\textwidth]{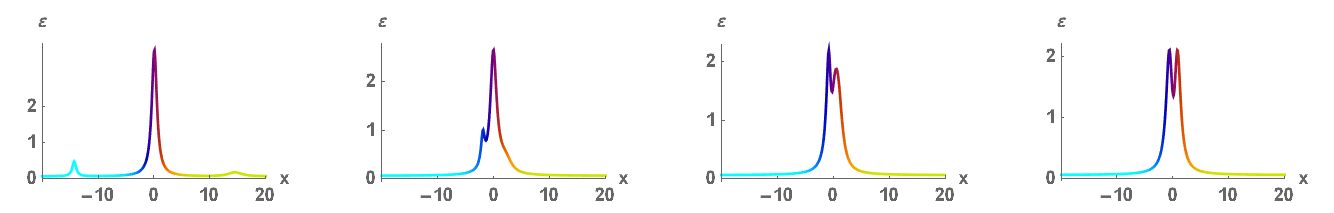}};
\node at (0,2*\a) {\includegraphics[width=\textwidth]{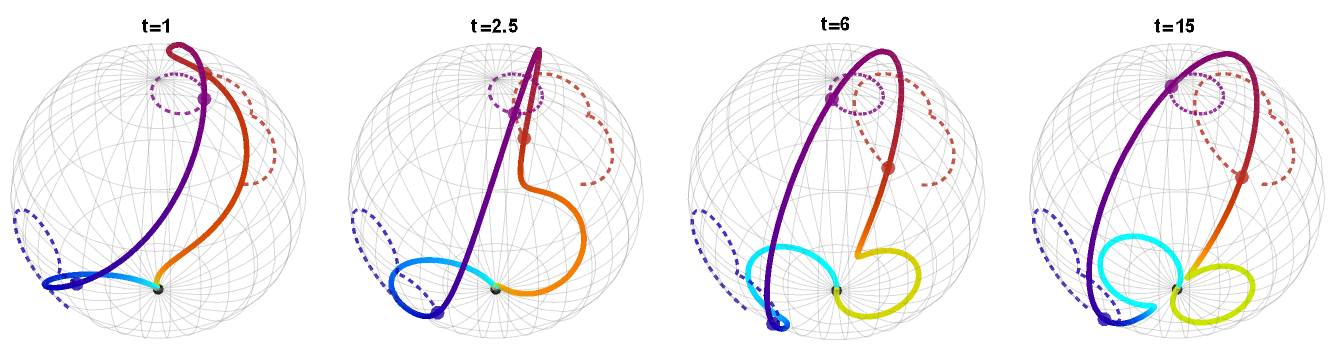}};
\node at (0,3*\a) {\includegraphics[width=\textwidth]{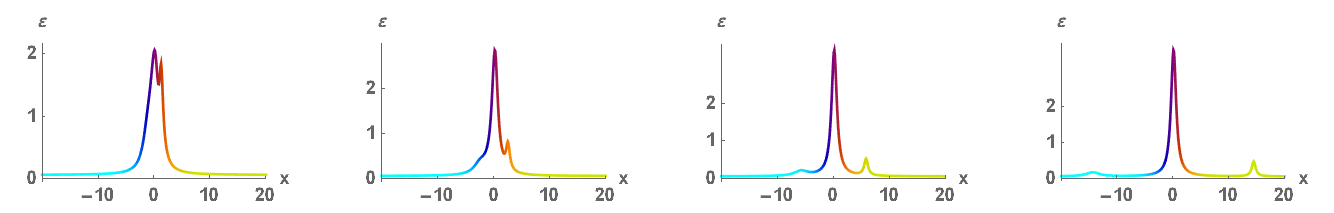}};
\node at (-0.26,3.7*\a-0.3) {\includegraphics[scale=1]{legend.png}};
\node at (0,3.7*\a+0.4) {$x$};
\foreach \x in {0,2,4}
{
\draw[gray] (-8.5,\b+\x*\a) -- (8.5,\b+\x*\a);
};

\foreach \x in {0,...,4}
{
\draw[gray] (-8.5+\x*\c,\b) -- (-8.5+\x*\c,\b+4*\a);
};

\foreach \y in {0,2}
{
\foreach \x in {0,...,3}
{
\fill[white] (-7.5+\x*\c,\b+\y*\a-0.5) rectangle (-4.5+\x*\c,\b+\y*\a-0.05);
}
};

\node at (-6.2,\b-0.3) {\scriptsize $t=-15$};
\node at (-6.4+\c,\b-0.3) {\scriptsize $t=-2$};
\node at (-6.4+2*\c,\b-0.3) {\scriptsize $t=-0.5$};
\node at (-6.4+3*\c,\b-0.3) {\scriptsize $t=0$};

\node at (-6.2,\b+2*\a-0.3) {\scriptsize $t=1$};
\node at (-6.3+\c,\b+2*\a-0.3) {\scriptsize $t=2.5$};
\node at (-6.4+2*\c,\b+2*\a-0.3) {\scriptsize $t=6$};
\node at (-6.4+3*\c,\b+2*\a-0.3) {\scriptsize $t=15$};

\end{tikzpicture}
\caption{The time evolution of the three-soliton solution with initial data \eqref{eq:initial_data_three_soliton}, following the lay-out of Fig. \ref{fig:two_soliton}. 
In addition,  the time evolution of the spin $\bm(x,t)$ at the three distinguished points $x=-1$ (blue), $x=0$ (purple) and $x=1$ (red) is traced by a dashed line, starting at time $t=-15$. This numerical solution 
 exhibits cusps in both $x$ and $t$ at $t=0$. 
}
\label{fig:three_soliton}
\end{figure}

The two-soliton solution in Fig. \ref{fig:two_soliton} has initial data \eqref{eq:initial_data_two_soliton}, and the time evolution of the spin $\bm$, which is visualized by the colored curves on the sphere, is complicated; however, the energy density plots are simpler, showing clearly the motion of two localized energy lumps that collide at $t=0$  
and re-emerge unchanged after the collision, as expected for solitons. The curves on the sphere show that, for $x\to \pm\infty$, the spin $\bm$ always aligns with the vaccum $\bmb$ and, if the solitons are far apart, the behavior of $\bm$ is well approximated by two circles corresponding to a sum to two one-soliton solutions,  as already discussed in Section \ref{sec:constraints1}. 

The three-soliton solution in Fig. \ref{fig:three_soliton} has initial data \eqref{eq:initial_data_three_soliton}, and it 
is similar to the two-soliton case, except that the spin configuration $\bm(x,t)$ can be describes by three circles when the solitons are far apart, and, in the special case we consider, one of these circles is stationary. 
Moreover, the spin $\bm$ at the origin $x=0$ makes two full precessions during the interaction, and it is located diametrically opposite to $\bmb$ at times $t=0$ and $t\to\pm\infty$; see the purple dotted line in Fig. \ref{fig:three_soliton}. 
It is interesting to note that there seem to be cusps in the time evolution of $\bm(x,t)$ at $x=\pm 1$ and $t=0$ as well as a spatial cusp at $x=0$ and $t=0$. Theorem~\ref{thm} thus seems to apply only to $t > 0$ and $t < 0$ but not to any time interval containing $t=0$; nevertheless, it seems that our numerical solution does not have any problems at $t=0$. We therefore believe that the HWM equation also has weak solutions with cusps. 

We computed the time evolution of many multi-soliton initial data obtained with random initial data and the numeric iteration procedure described in Section~\ref{sec:iteration}, and we did not find any cusp-like singularity in any of these solutions. 
However, we observed cusps like in Fig.~\ref{fig:three_soliton} also in the two-solition solution shown in Fig.~\ref{fig:two_soliton} (they become visible when tracing the time evolution of $\bm(x,t)$ at the points $x\approx \pm 1.45$). We thus believe that, while cusps exist for the HWM equation, they are rare, and the exact method described in Section~\ref{sec:exact} is useful to generate initial conditions leading to cusps.

\subsubsection{Numerical implementation}

\paragraph{Julia.} Using the Julia programming language \cite{bezanson2017} we have developed a program to evolve any form of initial data according to the HWM equation \eqref{hwm} along the lines of \cite{pelloni2000}, allowing us to check for soliton solutions by providing only an initial profile: due to numerical inaccuracies, the time evolution converges numerically only for initial conditions corresponding to $N$-soliton solutions.
Our program is a  low-level implementation of a complicated system of ODE's obtained by discretizing the HWM equation \eqref{hwm}, and typical running times are $10-120$ mins  on a regular laptop, depending on the desired accuracy and time interval. 
It turned out to be unnecessary to manually enforce the solution to remain on the unit sphere as time evolves (the HWM equation preserves the norm as shown in \eqref{eq:length_preserving}, but numerical inaccuracies could in principle spoil this property): our solutions had unit norm, $\bm^2=1$, with high accuracy ($\sim 10^{-3}$) on any time interval we used. The solutions displayed in Figs. \ref{fig:two_soliton}--\ref{fig:three_soliton} were checked using this implementation. 

\paragraph{Mathematica.} We used Mathematics to solve the second-order system \eqref{sjdotthm}--\eqref{ajddotthm} belonging to the rational case as well as \eqref{sjdotthm1}--\eqref{ajddotthm1} belonging to the trigonometric case. This runs very fast (typically $\ll 1 s$ on a regular laptop),  and it yields accurate results (numerical inaccuracies are invisible in our plots). The plots in Figs. \ref{fig:two_soliton}--\ref{fig:three_soliton} were produced with data from this implementation.

\section{Discussion}
\label{sec:discussion}
In this paper we discussed a new family of multi-soliton solutions of the half-wave maps equation,  and we showed that the time evolution of the spins and poles parametrizing these solutions is given by an integrable spin Calogero-Moser system. We also developed numerical methods to find explicit examples of these solitons.

It is interesting to note that our results parallel classical results on the Benjamin-Ono equation and its relation to the rational and trigonometric CM system \cite{polychronakos1995,stone2008,abanov2009}. 
In fact, there are two such relations which seem unrelated: first, one obtains the Benjamin-Ono equation as a hydrodynamic description of a CM system \cite{polychronakos1995,abanov2009}. 
Second, one can obtain exact multi-soliton solutions of the BO equation by a pole-ansatz and with poles satisfying the dynamics of the CM system \cite{chen1979}; 
however, the latter CM system has no direct physical interpretation since the poles move in the complex plane (rather than on the real line); 
moreover, the time evolution of the poles is given by first order equations, but these first order equations happen to be B\"acklund transformations of a CM system \cite{wojciechowski1982}. 
For the HWM equation, there are two corresponding relations to a spin CM system: first, the HWM equation is derived from an exactly solvable spin-chain model 
which can be obtained as a limit from a spin CM system using Polychronakos' freezing trick \cite{polychronakos1993}. 
Second, the dynamics of our pole solutions is described by a complexified version of this spin CM system which, again, 
has no obvious physical interpretation. Moreover, we do not obtain the spin CM system directly, but it arises from a system of first order equations which, 
as we show, is a B\"acklund transformation of this spin CM system. 

The results presented in this paper suggest several avenues for further investigation.

\begin{enumerate}

\item The main results in this paper, Theorems~\ref{thm} and \ref{thm1}, exclude 
 solutions with cusps in either $x$ or $t$. However, we observed such
cusps in our numerical experiments; see Fig.~\ref{fig:three_soliton}. The Camassa-Holm equation \cite{camassa1993} is known to admit peakons: soliton-like weak solutions with a discontinuous spatial derivative. It would be interesting to investigate the existence of analogous weak solutions in the HWM equation.

\item To our knowledge, the B\"acklund transformation 
 of the spin CM model that we found was not known before; it would be interesting to generalize it to other spin CM systems and study it systematically.

\item As the HWM equation can be interpreted as a classical continuum limit of the Haldane-Shastry spin chain, it is natural to investigate whether solitons exist for similar limits of related models. We propose to study the continuum dynamics of the elliptic spin CM system (and the corresponding hyperbolic, trigonometric, and rational systems via elliptic degeneration). Via freezing \cite{polychronakos1993}, such a continuum description would be expected to yield a classical hydrodynamic limit of the Inozemtsev spin chain \cite{inozemtsev1990connection} and thus provide an elliptic analog of the HWM equation.

\item The recently rediscovered partially anisotropic Haldane-Shastry model \cite{uglov1995trigonometric,lamers2018resurrecting} is another natural candidate to admit an interesting, integrable continuum description. The resulting 
soliton equation would expectedly be a $q$-deformation of the HWM equation and have solitons governed by Ruijsenaars-Schneider spin-pole dynamics. 

\item To our knowledge, a spin-pole ansatz was first used in \cite{gibbons1984generalisation} to obtain solution solutions of the Boomeron equation \cite{calogero1976coupled}. It would be interesting to investigate the applicability of such an ansatz to other models. The classical Heisenberg ferromagnet equation, $\bm_t=\bm\wedge\bm_{xx}$, is a natural place to start. While the soliton solutions of this model can be obtained through inverse scattering \cite{takhtajan1977,demontis2018b}, the motion of the poles associated with these solitons has not to our knowledge been studied before. The Heisenberg ferromagnet equation is gauge-equivalent to the nonlinear Schr\"{o}dinger equation \cite{zakharov1979}, whose solutions have poles that satisfy constrained CM dynamics \cite{hone1997}.  Thus, it would also be interesting to link the HWM equation to a nonlinear Schr\"{o}dinger-type equation. A natural candidate for such an equation is the Hilbert-nonlocal nonlinear Schr\"{o}dinger equation \cite{matsuno2000}.

\end{enumerate}

\appendix
\section{Proofs}\label{app:proofs} 
We give detailed proofs of the results presented in the main text, including computational details. 

Our arguments are such that they apply to the HWM equation on the real line (Section~\ref{sec2}) and the periodic HWM equation (Section~\ref{sec3}), using the special functions $\alpha(x)$ and $V(x)$ introduced in \eqref{alphaV} and \eqref{alphaV1}, respectively. 
Moreover, we generalize the results presented in the main text by allowing {\em complex}-valued solutions $\bm(x,t)$ of the HWM equation in \eqref{hwm} satisfying 
\begin{equation} 
\bm(x,t)^2 = \rho^2, 
\end{equation} 
where $\rho\in\C$ is an arbitrary complex constant,  and we allow for a more general pole ansatz
\begin{equation} 
\label{generalansatz}
\bm(x,t)=\bmb+\ii\sum_{j=1}^N \bs_j(t)\alpha(x-a_j(t))-\ii\sum_{j=1}^M \bt_j(t)\alpha(x-b_j(t))  
\end{equation}  
where 
\begin{equation} 
(a_j(t), \bs_j(t))\in\C_+\times \C^3 \quad (j=1,\ldots,N),\quad (b_k(t), \bt_k(t))\in\C_-\times \C^3\quad (k=1,\ldots,M) , 
\end{equation} 
with arbitrary non-negative integers $N,M$; the real-solutions case discussed in the main text corresponds to $\rho=1$ and 
\begin{equation} 
\label{realsolution} 
M=N; \quad b_j(t)=a_j(t)^*,\quad \bt_j(t)=\bs_j(t)^*\quad (j=1,\ldots,N). 
\end{equation} 
While the real-solution case is probably most interesting from a physics point of view, the more general complex solutions are interesting mathematically. 
Moreover, by using the short-hand notation introduced in \eqref{absr}, we get this generalization for free: the only difference is that we have 
\begin{equation} 
\label{absr1} 
(a_j,\bs_j,r_j)=\begin{cases}
(a_j,\bs_j,+) & (j=1,\ldots,N) \\
(b_{j-N},\bt_{j-N},-) & (j=N+1,\ldots N+M) 
\end{cases} 
\end{equation}  
instead of \eqref{absr} and $N+M$ instead of $2N$, and thus one can write \eqref{generalansatz} as
\begin{equation}\label{ansatzproof}
\bm(x,t)=\bmb+\ii\sum_{j=1}^{\cN} r_j\bs_j(t) \alpha(x-a_j(t)),\quad \cN\coloneqq N+M, 
\end{equation}
which is not more complicated in the general case. However, when inserting \eqref{absr1}, the special case  \eqref{realsolution} becomes simpler in that the equations obeyed by $(\bt_j,b_j)_{j=1}^M$ can be obtained by complex conjugation from the ones for $(\bs_j,a_j)_{j=1}^N$ and thus can be ignored.

In the following sections, we present our results and their proofs in a formal mathematical language, mirroring our more informal discussion in Section~\ref{sec2}.
 
\subsection{Constraints} 
\label{lemma1proof}
\begin{lemma}\label{normlemmaperiodic}
The function $\bm(x,t)$ in \eqref{generalansatz} satisfies $\bm(x,t)^2=\rho^2$ if and only if 
\begin{subequations}\label{normconditionsperiodic}
\begin{align}
\bs_j^2=0, \quad 
\bs_j\cdot \bigg(\ii \bmb-\ssum{k}{j}{N} \bs_k \alpha(a_j-a_k)   + \sum_{k=1}^M \bt_k \alpha(a_j-b_k) \bigg)=0 \label{con11} 
\end{align}
for $j=1,\ldots,N$, 
\begin{align}
\bt_j^2=0, \quad 
\bt_j\cdot \bigg(\ii \bmb+ \ssum{k}{j}{M} \bt_k \alpha(b_j-b_k)   - \sum_{k=1}^N \bs_k \alpha(b_j-a_k) \bigg)=0 \label{con22}
\end{align}
for $j=1,\ldots,M$, and 
\begin{equation}
\rho^2 =\bmb^2 + \kappa^2\bigg(\sum_{j=1}^N\bs_j-\sum_{j=1}^M\bt_j \bigg)^2.
 \end{equation}
\end{subequations}
\end{lemma}

\begin{proof} 
Using \eqref{ansatzproof}, we compute
\begin{align*}
\bm^2=\bmb^2+2\ii  \bmb\cdot \sum_{j=1}^{\cN}r_j\bs_j\alpha(x-a_j)-\sum_{j=1}^{\cN}\sum_{k=1}^{\cN}r_jr_k (\bs_j\cdot\bs_k)\alpha(x-a_j)\alpha(x-a_k) . 
\end{align*}
We evaluate the double sum using 
\begin{equation} 
\label{Id0} 
\alpha(x-a_j)^2=V(x-a_j)-\kappa^2
\end{equation} 
for $k=j$ and the identity
\begin{equation}
\label{Id1} 
\alpha(x-a_j)\alpha(x-a_k)=\alpha(a_j-a_k)\big(\alpha(x-a_j)-\alpha(x-a_k)\big)-\kappa^2
\end{equation}
for $k\neq j$. This gives
\begin{align*} 
\bm^2= \bmb^2+\kappa^2\sum_{j,k=1}^N r_j r_k (\bs_j\cdot\bs_k) - \sum_{j=1}^{\cN}\bs_j^2 V(x-a_j) +2\ii  \bmb\cdot \sum_{j=1}^{\cN}r_j\bs_j\alpha(x-a_j)\\
-2\sum_{j=1}^{\cN}\ssum{k}{j}{\cN}r_jr_k(\bs_j\cdot\bs_k)\alpha(a_j-a_k)\alpha(x-a_j), 
\end{align*} 
which implies that $\bm^2=\rho^2$ is equivalent to 
\begin{align}\label{constraintgeneral}
\bs_j^2=0,\quad \bs_j\cdot\bigg( \ii\bmb - \ssum{k}{j}{\cN} \bs_k\alpha(a_j-a_k) \bigg)=0 \quad (j=1,\ldots,\cN), 
\end{align} 
\begin{align} 
\rho^2 =  \bmb^2+\kappa^2\bigg(\sum_{j=1}^{\cN} r_j\bs_j \bigg)^2 . 
\end{align} 
Inserting \eqref{absr1} we obtain the result. 
\end{proof} 

\subsection{Spin-pole dynamics} 
\label{theorem1proof}
\begin{proposition}\label{poleansatzthm2}
The function $\bm(x,t)$ in \eqref{generalansatz} satisfies the HWM equation \eqref{hwm} if the constraints in Lemma~\ref{lemma1proof} and the following ordinary differential equations are fulfilled, 
\begin{subequations}
\begin{align}
\dot{\bs}_j=&-2\ssum{k}{j}{N}(\bs_j\wedge\bs_k)V(a_j-a_k)\quad (j=1,\ldots,N), \label{sjeq11} \\
\dot{\bt}_j=&-2\ssum{k}{j}{M}(\bt_j\wedge\bt_k)V(b_j-b_k)\quad (j=1,\ldots,M), \label{sjeq22} 
\end{align}
\end{subequations}
and 
\begin{subequations}
\begin{align}
\dot{a}_j \bs_j=& -\bs_j \wedge \bigg( \ii \bmb-\ssum{k}{j}{N} \bs_k \alpha(a_j-a_k) + \sum_{k=1}^M \bt_k \alpha(a_j-b_k)  \bigg) \quad (j=1,\ldots,N),  \label{ajeq11}\\
\dot{b}_j \bt_j=& \quad \bt_j \wedge \bigg( \ii \bmb  + \ssum{k}{j}{M} \bt_k \alpha(b_j-b_k) - \sum_{k=1}^N \bs_k \alpha(b_j-a_k) \bigg) \quad (j=1,\ldots,M). \label{ajeq22}
\end{align}
\end{subequations}
\end{proposition}

\begin{proof}
Inserting \eqref{ansatzproof} into \eqref{hwm} we find, on the left-hand side, 
\begin{equation}\label{mtproof}
\bm_t=\ii\sum_{j=1}^{\cN} r_j\big(\dot{\bs}_j \alpha(x-a_j)+ \bs_j \dot{a}_j V(x-a_j)\big)
\end{equation}
using $\alpha'(x-a_j)=-V(x-a_j)$. To compute the right-hand side of \eqref{hwm}, we use that
\begin{align*}
H\alpha(x-a_j)=-\ii r_j \alpha(x-a_j) \quad (j=1,\ldots,\cN)
\end{align*}
implying 
\begin{align*}
H\bm_x= -\sum_{k=1}^{\cN}  \bs_k V(x-a_k), 
\end{align*}
as explained in the main text (note that the arguments given there hold true even in the trigonometric case), and with \eqref{ansatzproof},  
\begin{align*}
\bm\wedge H\bm_x
=\sum_{j=1}^{\cN} (\bmb\wedge\bs_j) \alpha'(x-a_j)-\ii\sum_{j=1}^{\cN}\ssum{k}{j}{\cN} r_j(\bs_j\wedge\bs_k) \alpha(x-a_j)V(x-a_k). 
\end{align*}
We insert 
\begin{equation}
\label{Idkey} 
\alpha(x-a_j)\alpha'(x-a_k) =-\alpha(a_j-a_k)\alpha'(x-a_k) + \alpha'(a_j-a_k)\big(\alpha(x-a_j)-\alpha(x-a_k)\big) 
\end{equation}
obtained from \eqref{Id1} by differentiation with respect to the variable $a_k$,  and with  $\alpha'(x)=-V(x)$:  
\begin{align*}
\bm\wedge H\bm_x =- \sum_{j=1}^{\cN} (\bmb\wedge\bs_j) V(x-a_j)+\ii\sum_{j=1}^{\cN}\ssum{k}{j}{\cN} r_j(\bs_j\wedge\bs_k)\alpha(a_j-a_k)V(x-a_k)  \\
-\ii\sum_{j=1}^{\cN}\ssum{k}{j}{\cN} r_j(\bs_j\wedge\bs_k)  V(a_j-a_k)\big(\alpha(x-a_j)-\alpha(x-a_k)\big) .
\end{align*}
By using the antisymmetry of the cross product $\wedge$, $\alpha(a_j-a_k)=-\alpha(a_k-a_j)$, and swapping some of the summation indices $j\leftrightarrow k$ in the double sums, we can write this as 
\begin{align*}
\bm\wedge H\bm_x
= \sum_{j=1}^{\cN} (\bs_j\wedge\bmb) V(x-a_j)+\ii\sum_{j=1}^{\cN}\ssum{k}{j}{\cN} r_k(\bs_j\wedge\bs_k)\alpha(a_j-a_k)V(x-a_j)  \\
-\ii\sum_{j=1}^{\cN}\ssum{k}{j}{\cN} (r_j+r_k)(\bs_j\wedge\bs_k)  V(a_j-a_k)\alpha(x-a_j).
\end{align*}
This is equal to right-hand side in \eqref{mtproof} if and only if 
\begin{subequations}\label{theorem1eom}
\begin{align}
\dot{\bs}_j=&-\ssum{k}{j}{\cN}(1+r_jr_k)(\bs_j\wedge\bs_k)V(a_j-a_k), \\
\dot{a}_j \bs_j=&  -r_j\bs_j \wedge \sum_{j=1}^{\cN} \bigg( \ii \bmb- \ssum{k}{j}{\cN} r_k\bs_k \alpha(a_j-a_k)       \bigg). \label{dotajgeneral}
\end{align}
\end{subequations}
We insert \eqref{absr1} and obtain the result. 
\end{proof}

\subsection{Consistency of spin-pole dynamics}
\label{app:consistency} 
Equations \eqref{ajeq11}   and \eqref{ajeq22}  are only consistent if their right-hand sides are parallel with $\bs_j$ and $\bt_j$, respectively (otherwise, a contradiction would arise). As explained further below, consistency is guaranteed by the following basic result. 

\begin{lemma} 
\label{lem:bs}
Let $\bs\in\C^3$ be non-zero and such that $\bs^2=0$.  Then, for any $\bv\in\C^3$ satisfying $\bs\cdot\bv=0$, $\bs\wedge\bv$ is parallel with $\bs$: 
\begin{equation} 
\bs\wedge\bv =\frac{(\bs^*\wedge\bs)\cdot\bv}{\bs^*\cdot\bs} \bs . 
\end{equation} 
\end{lemma} 
(The proof is given in Appendix~\ref{app:bsproof}). 

Lemma~\ref{lem:bs} implies that, due to the constraints in  \eqref{con11} and \eqref{con22},  the right-hand sides of  \eqref{ajeq11}   and \eqref{ajeq22}  are parallel with $\bs_j$ and $\bt_j$, 
respectively, and that the latter two equations are equivalent to  
\begin{subequations} 
\begin{equation} 
\dot a_j= -\frac{\bs_j^*\wedge\bs_j}{\bs_j^*\cdot\bs_j} \cdot  \bigg( \ii \bmb-\ssum{k}{j}{N} \bs_k \alpha(a_j-a_k) + \sum_{k=1}^M \bt_k \alpha(a_j-b_k)  \bigg)\quad (j=1,\ldots,N) \label{ajeq11a}
\end{equation} 
and 
\begin{equation} 
\dot b_j= \frac{\bt_j^*\wedge\bt_j}{\bt_j^*\cdot\bt_j} \cdot  \bigg( \ii \bmb  + \ssum{k}{j}{M} \bt_k \alpha(b_j-b_k) - \sum_{k=1}^N \bs_k \alpha(b_j-a_k) \bigg)\quad (j=1,\ldots,M), \label{ajeq22a} 
\end{equation} 
\end{subequations} 
respectively. This makes manifest that the result in Proposition~\ref{theorem1proof} is consistent. 

\subsection{Relation to spin Calogero-Moser system}
\label{proofddotaj}

\begin{proposition} 
If the constraints in Lemma~\ref{normlemmaperiodic} hold true, then the first order equations in Proposition~\ref{poleansatzthm2} imply 
\begin{subequations} 
\begin{align} 
\ddot a_j = -\ssum{k}{j}{N} 2V'(a_j-a_k)\quad (j=1,\ldots,N), \\
\ddot b_j = -\ssum{k}{j}{M} 2V'(b_j-b_k)\quad (j=1,\ldots,M). 
\end{align} 
\end{subequations} 

\end{proposition} 
\begin{proof} 
We use the shorthand notation 
\begin{equation}\label{bjdefinition}
\bb_j\coloneqq \ii \bmb - \ssum{k}{j}{\cN}r_k\bs_k \alpha(a_j-a_k)
\end{equation}
allowing to write the constraints \eqref{constraintgeneral}  and the equations of motion \eqref{theorem1eom} as
\begin{align}
\bs_j^2=0, \quad 
\bs_j\cdot\bb_j=0 ,\quad 
\dot{\bs}_j=\ssum{k}{j}{\cN}(1+r_j r_k)(\bs_j\wedge\bs_k) \alpha'(a_j-a_k), \quad \bs_j\dot{a}_j=-r_j \bs_j\wedge \bb_j .
\label{alleqs} 
\end{align}
In the following, we only use these four equations and properties of the function $\alpha(x)$. We differentiate the last equation in \eqref{alleqs} with respect to time to obtain
\begin{equation}\label{bsjddotaj}
\bs_j \ddot{a}_j=-\dot{\bs}_j\dot{a}_j- r_j\dot{\bs}_j\wedge \bb_j- r_j \bs_j\wedge\dot{\bb}_j.
\end{equation}
To proceed, we use \eqref{bjdefinition} to compute
\begin{equation*}
\dot{\bb}_j=-\ssum{k}{j}{\cN} r_k\dot{\bs}_k \alpha(a_j-a_k)-\ssum{k}{j}{\cN}  r_k \bs_k \alpha'(a_j-a_k)(\dot{a}_j-\dot{a}_k). 
\end{equation*}
This allows us to write \eqref{bsjddotaj} as  
\begin{equation}\label{bsjddotaj1}
\bs_j \ddot{a}_j=\bc_{j,1}+\bc_{j,2} 
\end{equation}
with 
\begin{align}
\bc_{j,1}:=&\ssum{k}{j}{\cN} r_jr_k (\bs_j\wedge\dot{\bs}_k)  \alpha(a_j-a_k),  \label{bcj1definition} \\
\bc_{j,2}:=&-\dot{\bs}_j\dot{a}_j- r_j \dot{\bs}_j\wedge\bb_j +\ssum{k}{j}{\cN} r_j r_k(\bs_j\wedge\bs_k) \alpha'(a_j-a_k)(\dot{a}_j-\dot{a}_k). \label{bcj2definition}
\end{align}

We first compute $\bc_{j,1}$ using the third equation in \eqref{alleqs}:
\begin{align*}
\bc_{j,1}=&\ssum{k}{j}{\cN} r_jr_k \bs_j\wedge \Bigg(\ssum{l}{k}{\cN} (1+r_k r_l)(\bs_k\wedge\bs_l) \alpha'(a_k-a_l)  \Bigg) \alpha(a_j-a_k) \\
=& \ssum{k}{j}{\cN} \ssum{l}{k}{\cN}r_j(r_k+r_l) \big(\bs_j \wedge(\bs_k\wedge\bs_l)\big) \alpha(a_j-a_k)\alpha'(a_k-a_l) \\
=& \ssum{k}{j}{\cN}r_j(r_k+r_j) \big(\bs_j \wedge(\bs_k\wedge\bs_j)\big) \alpha(a_j-a_k)\alpha'(a_k-a_j) \\
&+ \ssum{k}{j}{\cN} \sssum{l}{j}{k}{\cN}r_j(r_k+r_l) \big(\bs_j \wedge(\bs_k\wedge\bs_l)\big) \alpha(a_j-a_k)\alpha'(a_k-a_l).
\end{align*}
Inserting $\bs_j\wedge(\bs_k\wedge\bs_j)=\bs_j^2\bs_k-(\bs_k\cdot\bs_j)\bs_j$ with $\bs_j^2=0$ and symmetrizing the second sum gives
\begin{align}
\bc_{j,1}=&- \ssum{k}{j}{\cN}(1+r_j r_k)(\bs_j\cdot\bs_k) \bs_j \alpha(a_j-a_k)\alpha'(a_k-a_j) \label{bcj1final}\\
&+\frac{1}2 \ssum{k}{j}{\cN} \sssum{l}{j}{k}{\cN}r_j(r_k+r_l) \big(\bs_j  \wedge(\bs_k\wedge\bs_l)\big)\big(\alpha(a_j-a_k)-\alpha(a_j-a_l)\big)\alpha'(a_k-a_l). \nonumber
\end{align}
We next compute $\bc_{j,2}$ in  \eqref{bcj2definition} by inserting the third equation in \eqref{alleqs}: 
\begin{align*}
\bc_{j,2}=&-\ssum{k}{j}{\cN} (1+r_j r_k)(\bs_j\wedge\bs_k)\alpha'(a_j-a_k) \dot{a}_j \\
& -\ssum{k}{j}{\cN} (r_j+r_k)((\bs_j\wedge\bs_k)\wedge\bb_j) \alpha'(a_j-a_k)       \\
&+\ssum{k}{j}{\cN} r_j r_k(\bs_j\wedge\bs_k) \alpha'(a_j-a_k)(\dot{a}_j-\dot{a}_k) \\
=&-\ssum{k}{j}{\cN} (\bs_j\wedge\bs_k)\alpha'(a_j-a_k) \dot{a}_j-\ssum{k}{j}{\cN} r_j r_k(\bs_j\wedge\bs_k)\alpha'(a_j-a_k)\dot{a}_k \\
& -\ssum{k}{j}{\cN} (r_j+r_k)((\bs_j\wedge\bs_k)\wedge\bb_j) \alpha'(a_j-a_k)  . 
\end{align*} 
We write as 
\begin{align*} 
\bc_{j,2}=&-\ssum{k}{j}{\cN} ((\dot{a}_j\bs_j)\wedge\bs_k)\alpha'(a_j-a_k)-\ssum{k}{j}{\cN} r_j r_k(\bs_j\wedge(\dot{a}_k\bs_k))\alpha'(a_j-a_k) \\
& -\ssum{k}{j}{\cN} (r_j+r_k)(\bs_j\wedge\bs_k)\wedge\bb_j \alpha'(a_j-a_k)
\end{align*}
and insert the last equation in \eqref{alleqs} to obtain 
\begin{align*}
\bc_{j,2}=&\ssum{k}{j}{\cN} r_j \big((\bs_j\wedge\bb_j)\wedge\bs_k\big) \alpha'(a_j-a_k)+\ssum{k}{j}{\cN} r_j (\bs_j \wedge(\bs_k\wedge\bb_k))\alpha'(a_j-a_k) \\
&- \ssum{k}{j}{\cN}(r_j+ r_k)\big((\bs_j\wedge\bs_k) \wedge\bb_j\big) \alpha'(a_j-a_k) \\
=&  \ssum{k}{j}{\cN}r_j\big( (\bs_j\wedge\bb_j)\wedge\bs_k+\bs_j\wedge(\bs_k\wedge\bb_k)-(\bs_j\wedge\bs_k)\wedge\bb_j \big)\alpha'(a_j-a_k) \\
&-\ssum{k}{j}{\cN} r_k \big((\bs_j\wedge\bs_k)\wedge\bb_j\big)\alpha'(a_j-a_k), 
\end{align*}
which can be computed with the  triple product identities $(\bx\wedge\by)\wedge\bz=(\bx\cdot\bz)\by-(\by\cdot\bz)\bx$ and $\bx\wedge(\by\wedge\bz)=(\bx\cdot\bz)\by-(\bx\cdot\by)\bz$ and $\bs_j\cdot\bb_j=0$: 
\begin{align}\label{bcj2equation1}
\bc_{j,2}=& \ssum{k}{j}{\cN}r_j\big\{ (\bs_j\cdot\bs_k)\bb_j -(\bb_j\cdot\bs_k)\bs_j+(\bs_j\cdot\bb_k)\bs_k-(\bs_j\cdot\bs_k)\bb_k \\ 
&-(\bs_j\cdot\bb_j)\bs_k+(\bs_k\cdot\bb_j)\bs_j\big\}\alpha'(a_j-a_k) \nonumber\\
&-\ssum{k}{j}{\cN} r_k \big( (\bs_j\cdot \bb_j)\bs_k - (\bs_k\cdot\bb_j)\bs_j\big)\alpha'(a_j-a_k) \nonumber\\
=& \ssum{k}{j}{\cN} \big(r_j (\bs_j\cdot\bs_k)(\bb_j-\bb_k) -   r_j(\bs_j\cdot(\bb_j-\bb_k))\bs_k +r_k(\bs_k\cdot(\bb_j-\bb_k))\bs_j\big)\alpha'(a_j-a_k). \nonumber
\end{align}
We use \eqref{bjdefinition}, $\alpha(x)=-\alpha(x)$, and the identity in \eqref{Idkey} for $x=a_l$ to compute 
\begin{align*}
(\bb_j-\bb_k)\alpha'(a_j-a_k)=&- \alpha'(a_j-a_k)\Bigg(\ssum{l}{j}{\cN} r_l\bs_l\alpha(a_j-a_l)-\ssum{l}{k}{\cN}r_l\bs_l\alpha(a_k-a_l)          \Bigg) \\
=&-(r_k\bs_k+r_j\bs_j)\alpha(a_j-a_k)\alpha'(a_j-a_k) \\
&+\sssum{l}{j}{k}{\cN}r_l\bs_l \alpha'(a_j-a_k)\big( \alpha(a_l-a_j)-\alpha(a_l-a_k)     \big) \\
=&-(r_k\bs_k+r_j\bs_j)\alpha(a_j-a_k)\alpha'(a_j-a_k)  \\
&-\sssum{l}{j}{k}{\cN}r_l\bs_l \big(  \alpha(a_k-a_j)-\alpha(a_l-a_j)      \big)\alpha'(a_l-a_k).
\end{align*}
Inserting this into \eqref{bcj2equation1} and using $\bs_j^2=0$ we find 
\begin{align*} 
\bc_{j,2}=&-\ssum{k}{j}{\cN} (1+r_jr_k) (\bs_j\cdot\bs_k)\bs_j\alpha(a_j-a_k)\alpha'(a_j-a_k) \\
&+\ssum{k}{j}{\cN} \sssum{l}{j}{k}{\cN}r_j r_l \big(  (\bs_j\cdot\bs_l)\bs_k -(\bs_j\cdot \bs_k)\bs_l \big) \big(  \alpha(a_k-a_j)-\alpha(a_l-a_j)      \big)\alpha'(a_l-a_k) \\
&-\ssum{k}{j}{\cN}  \sssum{l}{j}{k}{\cN}r_k r_l (\bs_k\cdot\bs_l)\bs_j \big(\alpha(a_k-a_j)-\alpha(a_l-a_j)\big)\alpha'(a_l-a_k). 
\end{align*}
Since $\alpha(-x)=-\alpha(x)$, the third summand is antisymmetric under the interchange of $k$ and $l$ and thus vanishes. 
Symmetrizing the second sum and inserting $(\bs_j\cdot\bs_l)\bs_k-(\bs_j\cdot\bs_k)\bs_l=\bs_j\wedge(\bs_k\wedge\bs_l)$, we obtain
\begin{align}\label{bcj2equation2}
\bc_{j,2}=&-\ssum{k}{j}{\cN} (1+r_jr_k) (\bs_j\cdot\bs_k)\bs_j\alpha(a_j-a_k)\alpha'(a_j-a_k) \\
&- \frac12\ssum{k}{j}{\cN} \sssum{l}{j}{k}{\cN}r_j (r_k+r_l)\big(\bs_j\wedge(\bs_k\wedge\bs_l)\big) \big(  \alpha(a_j-a_k)-\alpha(a_j-a_l)      \big)\alpha'(a_l-a_k).  \nonumber
\end{align}
Inserting  \eqref{bcj1final} and \eqref{bcj2equation2} into \eqref{bsjddotaj1}, the three-body terms cancel, and therefore 
\begin{equation*}
\ddot{a}_j\bs_j =- 2\ssum{k}{j}{\cN} (1+r_jr_k) (\bs_j\cdot\bs_k)\bs_j\alpha(a_j-a_k)\alpha'(a_j-a_k).
\end{equation*}
We cancel $\bs_j$ on both sides of the equation and insert $2\alpha(x)\alpha'(x)=V'(x)$ to obtain 
\begin{equation*} 
\ddot{a}_j =- \ssum{k}{j}{\cN} (1+r_jr_k) (\bs_j\cdot\bs_k)V'(a_j-a_k)
\end{equation*} 
which, by \eqref{absr1}, is equivalent to the result. 
\end{proof} 

\subsection{On spin Calogero-Moser systems and B\"acklund transformations}
\label{app:EOM}
The second-order differential equations for the poles $a_j$ and $b_j$ obtained in Appendix   \ref{proofddotaj} are remarkable since they decouple the dynamics of the two sets of variables $\{a_j,\bs_j\}_{j=1}^N$ and $\{b_j,\bt_j\}_{j=1}^N$: 
\begin{subequations} 
\label{eom} 
\begin{equation} 
\dot{\bs}_j=-2\ssum{k}{j}{N}\bs_j\wedge\bs_kV(a_j-a_k),\quad \ddot a_j = -\ssum{k}{j}{N} 4V'(a_j-a_k)
\end{equation} 
for $j=1,\ldots,N$, and 
\begin{equation} 
\dot{\bt}_j=-2\ssum{k}{j}{M}(\bt_j\wedge\bt_k)V(b_j-b_k),\quad \ddot b_j = -\ssum{k}{j}{M} 4V'(b_j-b_k)
\end{equation} 
\end{subequations} 
for $j=1,\ldots,M$. 
Moreover, as shown in Section~\ref{sec:CS_dynamics} in a representative special case, both these sets of equation are equal to the equations of motion of a spin CM model defined by the Hamiltonian
\begin{equation} 
H_{\mathrm{CM}} = \frac{1}{2}\sum_{j=1}^N p_j^2+ \sum_{1\leq j<k\leq N} 2  (\bS_j\cdot \bS_k) V(q_j-q_k)
\end{equation} 
and the Poisson brackets  \eqref{eq:poisson_brackets_rat}, and the latter is a special case of the spin CM system defined by the Hamiltonian \begin{equation}
H =  \frac{1}{2} \sum_{j=1}^N p_j^2 + \sum_{1\leq j<k\leq N} (\bv_j\cdot \bw_k)  (\bv_k\cdot \bw_j)V(q_j -q_k)
\end{equation}
and the Poisson brackets \eqref{eq:KBBT_Poisson}  \cite{gibbons1984generalisation,krichever1995spin}. 
These claims can be verified by straightforward generalizations of arguments given in the main text. 

It is natural to regard the first-order equations in Proposition~\ref{poleansatzthm2} as a B\"acklund transformation between two such spin CM models with different variable numbers $N$ and $M$, in generalization of a result due to Wojciechowski for standard CM models \cite{wojciechowski1982}. 

\section{On complex three-vectors squaring to zero}
\label{app:bs} 
In our soliton-ansatz,  we obtain spin degrees of freedom, $\bs_j\in\C^3$, which are non-zero and satisfy the contraint $\bs_j^2=0$. 
We collect and prove some basic properties about such complex three-vectors squaring to zero that we use. 

We recall that the set of all vectors in $\R^3$ with length 1 is denoted as $S^2$. 

\begin{lemma} 
\label{lem:sreps} 
Let $\bs\in\C^3$ be non-zero and such that $\bs^2=0$. Then the following statements hold true. 

\noindent (a) One can write   
\begin{equation} 
\label{sreps} 
\bs= s(\bn_1+\ii\bn_2)
\end{equation} 
with $s\in\C$ and $\bn_1,\bn_2\in S^2$ such $\bn_1\cdot\bn_2=0$, and this representation of $\bs$ is unique up to the following $\mathrm{U}(1)$ transformations, 
\begin{equation} 
\label{U1} 
(s,\bn_1,\bn_2)\to(s\ee^{\ii \alpha},\bn_1\cos\alpha+\bn_2\sin\alpha,-\bn_1\sin\alpha+\bn_2\cos\alpha) \quad (\alpha\in\R). 
\end{equation} 
(b) For any vector $\bv\in\C^3$, 
\begin{equation} 
\label{vreps} 
\bv = \tilde{v}_{1}\bs+\tilde{v}_{2}\bs^*+\tilde{v}_{3}\bs^*\wedge\bs; \quad \tilde{v}_{1}=\frac{\bs^*\cdot\bv}{\bs^*\cdot\bs},\quad \tilde{v}_{2}=\frac{\bs\cdot\bv}{\bs^*\cdot\bs},\quad \tilde{v}_{3} = \frac{(\bs\wedge\bs^*)\cdot\bv}{(\bs^*\cdot\bs)^2} . 
\end{equation} 
\end{lemma} 

\begin{proof} 
(a) Write $\bs=\re(\bs)+\ii\im(\bs)$ with $\re(\bs)$ and $\im(\bs)$ in $\R^3$, and note that $\bs^2=0$ is equivalent to $\re(\bs)^2=\im(\bs)^2$ and $\re(\bs)\cdot\im(\bs)=0$. This implies a unique representation of $\bs$ as in \eqref{sreps} with $s=|\re(\bs)|=|\im(\bs)|>0$, $\bn_1=\re(\bs)/s$ and $\bn_2=\im(\bs)/s$.  The invariance of \eqref{sreps} under the transformations \eqref{U1} follows from the obvious identity 
\begin{align*} 
s(\bn_1+\ii\bn_2) = \ee^{\ii\alpha} s\left((\bn_1\cos\alpha+\bn_2\sin\alpha)+\ii(-\bn_1\sin\alpha+\bn_2\cos\alpha)  \right) 
\end{align*}  
and the fact that $\bn_1'\coloneqq \bn_1\cos\alpha+\bn_2\sin\alpha$ and $\bn_2'\coloneqq-\bn_1\sin\alpha+\bn_2\cos\alpha$ both are in $S^2$ and satisfy $\bn_1'\cdot\bn_2'=0$, for arbitrary real $\alpha$. 

\noindent (b) The result in \eqref{vreps} follows since $\bs$, $\bs^*$ and $\bs^*\wedge\bs=2\ii|s|^2\bn_1\wedge\bn_2\coloneqq 2\ii|s|^3\bn_3$ are linearly independent; to get the formulas for $\tilde{v}_a$, $a=1,2,3$, 
multiply both sides in the first equation in \eqref{vreps} by $\bs^*$, $\bs$ and $\bs\wedge\bs^*$, respectively, and use that $\bs^2=(\bs^*)^2=0$, $\bs^*\cdot\bs=2|s|^2$,  
and $(\bx\wedge\by)\cdot(\by\wedge\bx)=(\bx\cdot\by)^2-\bx^2\by^2$. 

The result in Lemma~\ref{lem:sreps}(a) shows that each non-zero $\bs\in\C^3$ satisfying $\bs^2=0$ is associated with three vectors in $S^2$, namely $\bn_{1}$, $\bn_{2}$,  and $\bn_3=\bn_1\wedge\bn_2$, and these three vectors provide a right-handed  orthonormal basis: $\bn_j\cdot\bn_k=\delta_{jk}$ and $\bn_j\wedge\bn_k=+\epsilon_{jkl}\bn_k$ for $j,k=1,2,3$.  Moreover, due to the $\mathrm{U}(1)$-symmetry in \eqref{U1}, the basis vectors $\bn_1$ and $\bn_2$ are not unique but can be rotated in the plane spanned by them at the cost of changing the phase of $s$, whereas $\bn_3$ is uniquely determined by $\bs$. For this reason, we fix the directions $\bn_{j}\coloneqq \bn_{j,3}$ when solving the constraints \eqref{constraint} in our soliton ansatz. 
Moreover, as discussed in the main text, this direction $\bn_{j}$ has a clear physical meaning; see Fig. \ref{fig:single_soliton}.

\subsection{Proof of Lemma~\ref{lem:bs}}
\label{app:bsproof}
Since
\begin{equation*} 
\bs\cdot(\bs\wedge\bv)=0,\quad (\bs\wedge\bs^*)\cdot(\bs\wedge\bv)=\bs^2(\bs^*\cdot\bv)-(\bs\cdot\bv)(\bs^*\cdot\bs)=0, 
\end{equation*} 
this result is implied by Lemma~\ref{lem:bs}(b) and $\bs^*\cdot(\bs\wedge\bv)=(\bs^*\wedge\bs)\cdot\bv$.
\end{proof}

\section*{Acknowledgements}
\noindent
BKB and EL thank Jonatan Lenells for collaboration on a closely related project. 
RK thanks Istv\'{a}n Sz\'{e}cs\'{e}nyi for stimulating discussions.  
EL is grateful to Patrick G\'erard, Enno Lenzmann and Michael Stone for inspiring and helpful discussions. 
We thank Julien Roussillon for discussions during the initial stages of this work. 
BKB acknowledges support from the G\"oran Gustafsson Foundation and from the European Research Council, Grant Agreement No. 682537. The work of RK was supported by the grant ``Exact Results in Gauge and String Theories'' from the Knut and Alice Wallenberg foundation. 
EL acknowledges support from the Swedish Research Council, Grant No. 2016-05167, and by the Stiftelse Olle Engkvist Byggm\"astare, Contract 184-0573.

\bibliographystyle{unsrt}

\bibliography{BKL}

\end{document}